\documentclass[onecolumn]{IEEEtran}
\usepackage{amsmath,paralist,amsthm,comment,amssymb}
\usepackage{cite}
\usepackage{graphicx}
\usepackage{algorithm} 
\usepackage{newfloat,algcompatible}
\DeclareMathOperator{\supp}{supp}
\usepackage{bbm}
\newcommand{\mc}[1]{\mathcal{#1}}

\newtheorem{theorem}{Theorem}
\newtheorem{corollary}[theorem]{Corollary}
\newtheorem{lemma}[theorem]{Lemma}

\newtheorem{remark}[theorem]{Remark}

\newtheorem{definition}[theorem]{Definition}
\newcommand\numberthis{\addtocounter{equation}{1}\tag{\theequation}}

\begin{document}
\title{On the Local Equivalence of 2D Color Codes and Surface Codes with Applications}
\author{Arun~B.~Aloshious, %~\IEEEmembership{Member,}
        Arjun~Nitin~Bhagoji, %~\IEEEmembership{Princeton~University,}
        and~Pradeep~Kiran~Sarvepalli,~\IEEEmembership{Member, IEEE}%
\thanks{Arun~B.~Aloshious and Pradeep~Kiran~Sarvepalli are  with the Department
of Electrical  Engineering,  Indian Institute of Technology, Madras 600 036, India}% <-this % stops a space
\thanks{Arjun~Nitin~Bhagoji is with the Department
of Electrical  Engineering, Princeton University, NJ  USA.}% <-this % stops a space
\thanks{Some early results in the paper were presented at the 2015 International Symposium on Information Theory, Hong Kong.}
\thanks{This research was supported by a grant from the Centre for Industrial Consultancy \& Sponsored Research, IIT Madras. }
}%

%\markboth{Submitted to transactions on information theory} {Aloshious \MakeLowercase{\textit{et al.}}: On the Local Equivalence of 2D Color Codes and Surface Codes with Applications}



\maketitle
\begin{abstract}
In recent years, there have been many studies on local stabilizer codes. Under the assumption of translation and scale invariance Yoshida classified such codes. His result implies that  translation invariant 2D color codes are equivalent to copies of toric codes. Independently, Bombin, Duclos-Cianci, and Poulin showed that a local translation   invariant 2D topological stabilizer code is locally equivalent to a finite number of copies of Kitaev's toric code.  In this paper we focus on 2D topological color codes and relax the assumption of translation invariance. Using a linear algebraic framework we show that any 2D color code (without boundaries) is locally equivalent to two copies of a related surface code. This surface code is induced by color code. We study the application of this equivalence to the decoding of 2D color codes over the bit flip channel as well as the quantum erasure channel. We report the performance of the color code on the square octagonal lattice over the quantum erasure channel. Further, we provide explicit circuits that perform the transformation between 2D color codes and surface codes. Our circuits do not require any additional ancilla qubits. 

\end{abstract}


% \fi
\IEEEpeerreviewmaketitle

\section{Introduction}

\IEEEPARstart{I}{n}  recent years, there have been many studies on local stabilizer codes. Under the assumption of translation and scale invariance Yoshida classified stabilizer codes with local generators \cite{yoshida11}. His result implies that translation invariant 2D color codes are equivalent to copies of toric codes.
Independently, Bombin, Duclos-Cianci, and Poulin   showed that a local translation invariant 2D topological stabilizer code is locally equivalent to a finite number of copies of Kitaev's toric code \cite{bombin12}. At first sight these results are somewhat surprising 
because topological color codes in 2D can implement the entire Clifford group transversally \cite{bombin06} while surface codes cannot implement the Clifford group transversally. 
Nonetheless, it turns out that color codes can be transformed to surface codes locally.

It is natural to ask whether this equivalence holds if the requirement of translational symmetry is relaxed. In this paper we focus on 2D topological color codes (without boundaries) and relax the assumption of translation invariance. We show local equivalence between arbitrary 2D color codes and copies of surface codes and study the applications of this equivalence in detail. Using a linear algebraic framework we show that any 2D color code is locally equivalent to two copies of a related surface code. This surface code is induced by the color code. This paper expands upon our previous work \cite{bhagoji15} and explores the implications of this equivalence in detail.

The equivalence between color codes and copies of surface codes has important and useful consequences. For instance, one could use this local equivalence to come up with decoding algorithms for color codes in terms of the decoding algorithms of the surface codes. Bombin et al. used this idea to decode translation invariant color codes \cite{bombin12}. Our result enables this type of decoding for color codes that are not translation invariant. Delfosse also studied the decoding of color codes via projection onto surface codes \cite{delfosse14}. We pursue this application in great detail by studying the decoding of color codes over both the bit (phase) flip channel and the quantum erasure channel. 

The local equivalence of color codes to copies of surface codes implies that some of the benefits of color codes can be ported to surface codes and may find application in code switching \cite{nautrup17}. However, applications for fault tolerant quantum computation need detailed study. The benefits of code transformation must be weighed against the overheads incurred in the transformation between codes. To understand this trade-off we study the circuits for transforming between the two types of codes. We summarise our main contributions:
\begin{compactenum}[i)]
	\item We show that a 2D color code is locally equivalent to two copies of a surface code. Our approach uses elementary linear algebra and coding theoretic methods. Our map is bijective and does not require any ancilla qubits. These results were presented by some of us previously in \cite{bhagoji15}.
	\item We study an application of this result to the decoding of 2D color codes over the bit (phase) flip channel. 
	\item We study the performance of color codes on the quantum erasure channel. Application of the local equivalence between the color codes and surface codes to erasure decoding has not been explored earlier.
	\item We give explicit circuits for the transformation of color codes into surface codes. 
\end{compactenum}
We also study how  the error model on the color code maps to the surface codes. This could be of interest when designing decoders for the color code.

We briefly review previous work and compare it with our contributions in this paper.  Yoshida studied stabilizer codes with local stabilizers, under the assumption of translational   and scale invariance \cite{yoshida11}. 
The latter symmetry means that the number of logical qubits does not depend on system size. 
His result led to a classification of such codes in terms of the geometry of the logical operators of the code. 
A consequence of this result is that translation invariant color codes are locally equivalent to copies of surface codes because their logical operators have the same geometry. 

Taking a complementary approach and focusing on local properties Bombin et al. showed that 2D translational invariant 
stabilizer codes and some subsystem codes can be mapped to a finite number of copies of the toric code \cite{bombin12}. 
They also studied an application of this result to decoding of color and subsystem codes. 
Rigorous proofs of the results in \cite{bombin12} were supplied in \cite{bombin14}. 

More recently, Haah showed that translation invariant Calderbank-Shor-Steane (CSS) codes can also be mapped to copies of toric code \cite{haah16}. His approach views  translation invariant codes as maps over free modules and makes heavy use of commutative algebra.

The previous three approaches require translation invariance which we do not impose. Further, we also take a simpler approach compared to them. Both Yoshida and Haah do not study decoders for color codes while we investigate the performance of decoders for the color codes. Bombin et al. study a soft decision decoder for color codes, while we study a hard decision decoder \cite{bombin12}. They report a threshold of $9\% $ while we obtain a threshold of $ 5.3\%$ for the color code on the square octagonal lattice. This difference is due to the decoder and not due to the map. 
A further difference is that the map by Bombin et al. can lead to multiple copies of Kitaev's toric code while our map always leads to two copies.

A related result in this context is that of Delfosse who projected 2D color codes onto surface codes \cite{delfosse14}.
His work was motivated by the problem of decoding color codes and does not explicitly discuss the issue of local equivalence. In this approach the bit flip errors are decoded by projecting onto one set of surface codes and the phase flip errors are decoded by projecting onto another set of codes. In contrast, our maps as well as those of Bombin et al.
use only one set of surface codes to decode both types of errors. For the color code on the square octagonal lattice  Delfosse's approach gives a threshold of about~7.6\% \cite{stephens14} while for the hexagonal lattice it gives a threshold of about 8.7\% \cite{delfosse14}. 
Decoding of color codes with boundaries using this approach was studied in \cite{stephens14}.
Other approaches are known for decoding color codes \cite{ps2012,wang09,landahl11} which decode directly on the color codes. 

Following a preliminary version of our paper, we became aware of the work by Kubica, Yoshida, and Pastawski \cite{kubica15}. They mapped color codes onto toric codes for all dimensions $D\geq 2$ in contrast to the present work which is restricted to 2D. 
Although in 2D, for color codes without boundaries, their result is similar to ours, there are substantial differences in the approach. First, they map the color code onto two different surface codes, while we map onto two copies of the same surface code. Second, we use linear algebra to study these equivalences, whereas their approach is based on algebraic topology. Third, we study the application of these maps for decoding. We also provide explicit circuits for transformation. 

Our paper is structured as follows. First we review the necessary background in Section~\ref{sec:bg}. Then, in Section~\ref{sec:map}, we derive the map between color codes and surface codes. In Section~\ref{sec:decoding}, we study an application of this map to decoding over the bit (phase) flip channel and the quantum erasure channel. In Section~\ref{sec:ckts} we provide efficient circuits for transforming an arbitrary 2D color code into surface codes. We then conclude with a brief discussion on the scope for future work.
\section{Preliminaries} \label{sec:bg}
We assume that the reader is familiar with stabilizer codes \cite{calderbank97,gottesman97} and topological codes \cite{kitaev03}, see also \cite{lidar13}. %bravyi98
The Pauli group on $n$ qubits is denoted $\mc{P}_n$. We denote the vertices of a graph $\Gamma$ by $\mathsf{V}(\Gamma)$, and the edges by $\mathsf{E}(\Gamma)$. The set of edges incident on a vertex $v$ is denoted as $\delta(v)$ and the edges in the boundary of a face by $\partial(f)$. Assuming that $\Gamma$ is embedded on a suitable surface we use $\mathsf{F}(\Gamma)$ to denote the faces of the embedding and do not always make an explicit reference to the surface. Consider a graph $\Gamma$ with qubits  placed on the edges. A surface code on  $\Gamma$ is a stabilizer code  whose stabilizer $S$ is given by 
\begin{align}
	S &= \langle A_v , B_f \mid v\in \mathsf{V}(\Gamma), f\in \mathsf{F}(\Gamma) \rangle, \label{eq:stab-sc}
\end{align}
where $A_v = \prod_{e\in \delta(v) }X_e \mbox{ and } B_f =\prod_{e\in \partial(f)}Z_e$. The Pauli group on the qubits of a surface code is denoted as $\mc{P}_{\mathsf{E}(\Gamma)}$. A 2-colex is a trivalent, 3-face-colorable complex. 
A color code is defined on a 2-colex by attaching qubits to every vertex 
and defining the stabilizer $S$ as 
\begin{eqnarray}
	S = \langle B_f^X, B_f^Z \mid v\in \mathsf{F}(\Gamma) \rangle \mbox{ where } B_f^\sigma = \prod_{v\in f }\sigma_v. \label{eq:stab-tcc}
\end{eqnarray}
We denote the Pauli group on these qubits as $\mc{P}_{\mathsf{V}(\Gamma)}$; the $c$-colored faces of $\Gamma$ by $\mathsf{F}_c(\Gamma)$ and the $c$-colored edges of $\Gamma$ by $\mathsf{E}_c(\Gamma)$. We restrict our attention to 2-colexes which do not have boundaries or multiple edges (the surface codes could contain multiple edges though). This is not a severe restriction because a 2-colex with multiple edges can be modified to another 2-colex without such edges but encoding the same number of qubits and possessing the same error correcting capabilities (in terms of distance). All embeddings are assumed to be 2-cell embeddings i.e. faces are homeomorphic to unit discs.  

There are four types of topological charges on a surface code:
\begin{compactenum}[i)]
	\item electric charge (denoted $\epsilon$) localized on the vertices
	\item magnetic charge (denoted $\mu$) living on the plaquettes
	\item the composite electric and magnetic charge denoted $\epsilon \mu$ which resides on both the plaquettes and vertices \item the vacuum (trivial charge) denoted $\iota$
\end{compactenum}

Of these, only two charges are independent. We shall take this pair to be the electric and magnetic charges. A charge composed with another charge of the same type gives the vacuum i.e. $c\times c =\iota$. The electric charges are created by $Z$-type errors and magnetic charges by $X$-type errors on the surface code. 

On a color code, the topological charges live on the faces. In addition to being electric and/or magnetic, they also carry a color depending on which face they are present. Let us denote the electric charge on a $c$-colored face as $\epsilon_c$, the magnetic charge as $\mu_c$ and the composite charge as $\epsilon_c\mu_{c}$. The electric charges are not all independent \cite{bombin06}. Any pair (two out of three colors) of them can be taken as the independent set of electric charges. Similarly, only two magnetic charges are independent. As in surface codes, electric (magnetic) charges are created by $Z$ ($X$) errors on the color code.

A hopping operator is any element of the Pauli group that moves the charges. On a surface code, we can move the electric charges from one vertex to another by means of a $Z$-type Pauli operator. We denote by $H_{u\leftrightarrow v}^{\epsilon}$ the operator that moves $\epsilon$ from vertex $u$ to $v$ and vice versa. If we consider the magnetic charges, then the movement can be accomplished by means of an $X$-type Pauli operator. The operator that moves a magnetic charge from face $f$ to $f'$ (or vice versa) is denoted by $H_{f\leftrightarrow f'}^{\mu}$. Elementary hopping operators are those which move charges from one vertex to an adjacent vertex or from one plaquette to an adjacent plaquette. Let $e=(u,v)$ be the edge incident on the vertices $u$, $v$. We denote the elementary hopping operator along $e$ as $H_e^{\epsilon}$, where $ H_e^{\epsilon} = Z_{e}$. It is a specific realization of $H_{u\leftrightarrow v}^{\epsilon}$. Similarly, the elementary operator that moves $\mu $ across $e$ is denoted as $H_e^\mu$. Let $e$ be the edge shared by the faces $f$ and $f'$, then $H_{f\leftrightarrow f'}^{\mu}$ can be realized by $H_e^{\mu }$ where $H_e^{\mu }=X_{e}$. Observe that $H_{u\leftrightarrow v}^{\epsilon}$ and $H_{f\leftrightarrow f'}^{\mu}$ anti-commute when they act along the same edge, while operators for the same type of charges commute. In general, $H_{u\leftrightarrow v}^{\epsilon}$ and $H_{f\leftrightarrow f'}^{\mu}$ commute if and only if they cross an even number of times. 

Similarly, we can define hopping operators for color codes. Let $f, f' \in \mathsf{F}_c(\Gamma)$ be two plaquettes connected by an edge $(u,v)$ where $u\in f$
and $v\in f'$. Then $H_{f\leftrightarrow f'}^{\epsilon_c}$ and $H_{f\leftrightarrow f'}^{\mu_c}$ are the operators that move $\epsilon_c$ and $\mu_c$ from $f$ to $f'$. A realization of these operators along $(u,v)$ is $H_{u,v}^{\epsilon_c} = Z_{u}Z_{v} $ and $H_{u,v}^{\mu_c}=X_{u}X_{v}$. An element of the stabilizer can be viewed as a combination of hopping operators which move a charge around and bring it back to the original location. Since this movement cannot be detected, we can always adjoin an element of the stabilizer to the hopping operators. 
\section{Mapping a color code to two copies of a surface code}\label{sec:map}
\subsection{Color codes to surface codes---Constraints}\label{ssec:constraints}
Our goal is to find a map between a color code and some related surface codes. We shall denote this map by $\pi$ for the rest of the paper. 
We shall first describe the construction of $\pi$ in an informal fashion, emphasizing the principles underlying the map, and then rigorously justify all the steps. The key observation, due to \cite{bombin12}, is that there are four types of charges on a surface code and sixteen types of charges on a color code. This is the starting point for relating the color code to surface codes. The two pairs of independent charges on the color code, i.e. $\{\epsilon_c, \mu_{c'} \}$ and $\{\epsilon_{c'}, \mu_{c} \}$, suggest that we can decompose the color code into a pair of toric codes by mapping $\{\epsilon_c, \mu_{c'} \}$ charges onto one toric code and $\{\epsilon_{c'}, \mu_{c} \}$ onto another. However, charge ``conservation'' is not the only constraint. We would like a map that preserves in some sense the structure of the color code and allows us to go back and forth between the color code and the surface codes. We shall impose some conditions on this map keeping in mind that we would like to use it in the context of decoding color codes. 

First, observe that the electric charges on the surface codes live on the vertices while the magnetic charges live on the plaquettes. But, if we consider the pair of charges $\{\epsilon_c, \mu_{c'} \}$, they both live on plaquettes---one on the $c$-colored plaquettes and another on $c'$-colored plaquettes. A natural way to make the association to a surface code is to contract all the $c$-colored plaquettes in the embedding of $\Gamma$. This will give rise to a new graph $\tau_c(\Gamma)$. We can now place the charges $\epsilon_c$ and $\mu_{c'}$ on the vertices and plaquettes of $\tau_c(\Gamma)$ respectively. Similarly, the pair of charges $\{\mu_{c}, \epsilon_{c'} \}$ can live on the vertices and plaquettes of {\em another} instance of $\tau_c(\Gamma)$. We impose the following (desirable) constraints on the map $\pi$. It must be 
{(i) linear,
	(ii) invertible,
	(iii) local,
	(iv) efficiently computable,
	(v) preserve the commutation relations between the (Pauli) error operators on $\mathsf{V}(\Gamma)$ i.e. $\mc{P}_{\mathsf{V}(\Gamma)}$,
	and (vi) consistent in the description of the movement of charges on the color code and surface codes.
}These constraints are not necessarily independent and in no particular order. It is possible to relax some of the constraints above. 
\subsection{Deducing the map---A linear algebraic approach }
The maps proposed in \cite{bombin12} are based on the following ideas: i) conservation of topological charges ii) identification of the hopping operators and iii) preserving the commutation relations between the hopping operators.
These ideas are central to our work as well. However, we take a simpler linear algebraic approach to find the map.

Suppose we have a 2-colex $\Gamma$. Then, upon contracting all the $c$-colored faces including their boundary edges, we obtain another complex. We denote this operation as $\tau_c$ and the resulting complex as $\tau_c (\Gamma)$ (see Fig.~\ref{fig:tau}). We suppress the subscript if the context makes it clear and just write $\tau$. There is a one-to-one correspondence between the $c$-colored faces of $\Gamma$ and the vertices of $\tau(\Gamma)$, so we can label the vertices of $\tau(\Gamma)$ by $f\in \mathsf{F}_c(\Gamma)$. We also label them by $\tau(f)$ to indicate that the vertex was obtained by contracting $f$. Similarly, the edges of $\tau(\Gamma)$ are in one-to-one correspondence with the $c$-colored edges of $\Gamma$, so an edge $\tau(\Gamma)$ is labeled the same as the parent edge $e=(u,v)$ in $\Gamma$. The faces which are not in $\mathsf{F}_c(\Gamma)$ are mapped to faces of $\tau(\Gamma)$. Therefore, we label the faces as 
$f$ or more explicitly as $\tau(f)$, where $f\not\in \mathsf{F}_c(\Gamma) $. Thus, the complex $\tau(\Gamma)$ has the vertex set $\mathsf{F}_c(\Gamma)$, edge set $\mathsf{E}_c(\Gamma)$ and faces $\mathsf{F}_{c'}(\Gamma)\cup \mathsf{F}_{c''}(\Gamma)$. Since every vertex $v$ in $\Gamma$ has a unique $c$-colored edge incident on it, we can associate to it an edge in $\tau(\Gamma)$ as $\tau(v)$. 

\begin{figure}
	\centering
	\includegraphics[scale=0.7]{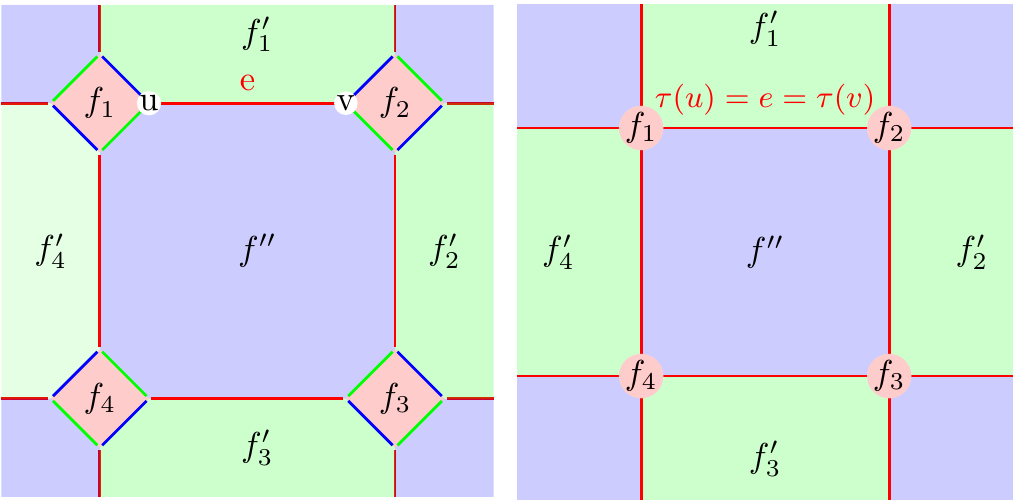}\label{fig:sc}
	\caption{Illustrating the contraction of a color code via $\tau_c$ and the resultant surface code. Only portions of the codes are shown. The $c$-colored faces are vertices in $\tau_c(\Gamma)$. The faces $f \not\in \mathsf{F}_c(\Gamma)$ remain faces in $\tau_c(\Gamma)$ and are also labeled $f$ in $\tau_c(\Gamma)$, while the $c$-colored edge $e=(u,v)$ in $\Gamma$ is mapped to an edge in $\tau_c(\Gamma)$, so we retain the label $e$. Every vertex in $\Gamma$ is incident on a unique $c$-colored edge, so we can also extend $\tau_c$ to vertices $u$, $v$ and edges unambiguously by defining $\tau_c(u)=\tau_c(v)=\tau_c(u,v)=e$.}\label{fig:tau}
\end{figure}

Now, each $c$-colored face in $\Gamma$ can host $\epsilon_c$ and $\mu_c$. With respect to $\tau_c(\Gamma)$, they both reside on the vertices of $\Gamma_c$. So we shall place them on two different copies of $\tau_c(\Gamma$) denoted $\Gamma_1$ and $\Gamma_2$. Then, the charges $\epsilon_c$ and $\mu_c$ will play the role of an electric charge on $\Gamma_1$ and $\Gamma_2$, respectively. So, we shall make the identification $\epsilon_c \equiv \epsilon_1$ and $\mu_c\equiv \epsilon_2$. The associated magnetic charges on $\Gamma_i$ will have to reside on $\mathsf{F}(\Gamma_i)$. Possible candidates for these charges must come from $\epsilon_{c'}$, $\epsilon_{c''}$ and $\mu_{c'}$, $\mu_{c''}$. The following lemma addresses these choices. 

\begin{lemma}[Charge mapping] \label{lm:charge-pairing}
	Let $c,c',c''$ be three distinct colors. Then, $\{\epsilon_c,\mu_{c'}\}$ and $\{ \epsilon_{c'},\mu_c \}$ are permissible pairings of the charges so that the color code on $\Gamma$ can be mapped to a pair of surface codes on $\Gamma_i=\tau_c(\Gamma)$. In other words, $\epsilon_1\equiv \epsilon_c$, $\mu_1\equiv \mu_{c'}$, $\epsilon_2\equiv \mu_c$ and $\mu_2\equiv \epsilon_{c'}$, where $\epsilon_i$ and $\mu_i$ are the electric and magnetic charges of the surface code on $\Gamma_i$.
\end{lemma}
\begin{proof} First, observe that operators that move the electric charges $\epsilon_c$ and $\epsilon_{c'}$ are both $Z$-type, therefore they will always commute. This means that if $\epsilon_c$ is identified with the electric charge on a surface code, $\epsilon_{c'}$ cannot be the associated magnetic charge. That leaves either $\mu_{c}$ and $\mu_{c'}$. Of these, observe  that any operator that moves $\mu_{c}$ will always overlap with any operator that moves $\epsilon_c$ an even number of times. Therefore, this leaves only $\mu_{c'}$. The operators that move $\epsilon_c$ and $\mu_{c'}$ commute/anti-commute when they overlap an even/odd number of times just as the electric and magnetic charges of a surface code justifying the association $\epsilon_1\equiv \epsilon_c$ and $\mu_1\equiv \mu_{c'}$. A similar argument shows the validity of the equivalence $\epsilon_2\equiv \mu_c$ and $\mu_2\equiv \epsilon_{c'}$. 
\end{proof}
Let $\Gamma$ have $n$ vertices and $F_c$ faces of color $c$. Then, $\Gamma_i$ has $F_c$ vertices, $n/2$ edges and $F_{c'}+F_{c''}$ faces. Together $\Gamma_1$ and $\Gamma_2$ have $n$ qubits. We desire that $\pi$ accurately reflect the movement of the independent charges on the color code and the surface codes. So, $\pi$ must map the hopping operators of the charges of the color code on $\Gamma$ to the hopping operators of the surface code on $\Gamma_i$. As mentioned earlier,
$H_{u,v}^{\epsilon_c}$ moves electric charges on $c$-colored plaquettes and 
$H_{u,v}^{\epsilon_{c'}}$ electric charges on $c'$-colored plaquettes. But, although these operators may appear to be independent, due to the structure of the color code they are not. A $c''$-colored plaquette on the color code is bounded by edges whose color alternates between $c$ and $c'$. The $Z$-type stabilizer associated to this plaquette, i.e. $B_f^Z$, can be viewed as being composed of $H_{u,v}^{\epsilon_c}$ hopping operators that move $\epsilon_c$, in which case we would expect to map $B_f^Z$ onto $\Gamma_1$. But, $B_f^Z$ can also be viewed as 
being composed of $H_{u,v}^{\epsilon_{c'}}$. Thus, we see that there are two possible combinations of hopping operators that give the same plaquette stabilizer; one composed entirely of hopping operators of $c$-colored charges and the other of hopping operators of $c'$-colored charges. This suggests that there are dependencies among the hopping operators and some of them, while ostensibly acting on only one kind of charge, could still be moving the other type of charges. However, the overall effect on the other charge must be trivial, i.e. it must move the charge back to where it started. A similar argument can be made for $B_f^X$ which moves the magnetic charges. The next lemma makes precise these dependencies. 

\begin{lemma}[Dependent hopping operators]\label{lm:dep-ops}
	Let $f \in \mathsf{F}_{c''}(\Gamma)$ and ${1},\ldots, {2{\ell_f}}$ be the vertices in its boundary so that $(v_{2i-1},v_{2i}) \in \mathsf{E}_c(\Gamma)$, $(v_{2i},v_{2i+1}) \in \mathsf{E}_{c'}(\Gamma)$ for $1\leq i\leq {\ell_f}$ and $2{\ell_f}+1 \equiv 1$. If $\pi$ is invertible, then $\pi(B_f^\sigma)\neq I$ and there are $4{\ell_f}-2$ independent elementary hopping operators along the edges of $f$.
\end{lemma}
\begin{proof}
	The stabilizer generator $B_f^Z$ is given as 
	\begin{align}
		B_f^Z& = \prod_{i=1}^{2{\ell_f}}Z_{v_i} = \prod_{i=1}^{\ell_f} Z_{v_{2i-1}} Z_{v_{2i}}=Z_{v_1}Z_{v_{2{\ell_f}}}\prod_{i=1}^{{\ell_f}-1} Z_{v_{2i}} Z_{v_{2i+1}}\\
		&=\prod_{i=1}^{\ell_f} H_{v_{2i-1}, v_{2i}}^{\epsilon_c} =H_{v_1, v_{2{\ell_f}}}^{\epsilon_{c'}}\prod_{i=1}^{{\ell_f}-1} H_{v_{2i}, v_{2i+1}}^{\epsilon_{c'}}.
	\end{align}
	We see that $B_f^Z$ can be expressed as the product of ${\ell_f}$ hopping operators of type $H_{u,v}^{\epsilon_c}$ or type $H_{u,v}^{\epsilon_{c'}}$. Further, we have
	\begin{eqnarray}
		\pi(B_f^Z)& =&\prod_{i=1}^{\ell_f} \pi(H_{v_{2i-1}, v_{2i}}^{\epsilon_c}) =\pi( H_{v_1, v_{2{\ell_f}}}^{\epsilon_{c'}})\prod_{i=1}^{{\ell_f}-1} \pi(H_{v_{2i}, v_{2i+1}}^{\epsilon_{c'}})\nonumber
	\end{eqnarray}
	
	If $\pi(B_f^Z)=I$, then $\ker(\pi)\neq I$ which means that $\pi $ is not invertible and it would not be possible to preserve the information about the syndromes, as $\pi(B^Z_f)$ would commute with all the error operators. So, we require that $\pi(B_f^Z)\neq I$. This means that only one of these hopping operators is dependent and there are at most $2{\ell_f}-1$ independent hopping operators. The linear independence 	of the remaining $2{\ell_f}-1$ operators can be easily verified by considering their support. Similarly, $B_f^X$ also implies that there are another $2{\ell_f}-1$ independent hopping operators, giving us $4{\ell_f}-2$ in total. 
\end{proof}

We are now ready to define the action of $\pi$ on elementary hopping operators. Without loss of generality we can assume if $f\in \mathsf{F}_{c''}(\Gamma)$ has $2{\ell_f}$ vertices, then the dependent hopping operators of $f$ are $H_{v_1,v_{2{\ell_f}} }^{\epsilon_{c'}}$ and $H_{v_{2m_f},v_{2m_f+1} }^{\mu_{c'}}$ i.e. $Z_{v_1}Z_{v_{2{\ell_f}}}$ and $X_{v_{2m_f}}X_{v_{2m_f+1}}$, where $1\leq m_f\leq {\ell_f}$ and $2{\ell_f}+1\equiv 1$. 

\begin{lemma}[Elementary hopping operators]\label{lm:hopping}
	Let $f, f' \in \mathsf{F}_c(\Gamma)$ where the edge $(u,v)$ is incident on $f$ and $f'$. Then, the following choices reflect the charge movement on $\Gamma$ onto the surface codes on $\Gamma_i$.
	\begin{eqnarray}
		\pi(H_{u, v}^{\epsilon_c}) & = & \left[Z_{\tau(u)}\right]_1 = \left[Z_{\tau(v)}\right]_1 \label{eq:elec-hopper1}\\
		\pi(H_{u , v}^{\mu_c}) & = & \left[Z_{\tau(u)} \right]_2= \left[Z_{\tau(v)}\right]_2,\label{eq:mag-hopper1}
	\end{eqnarray}
	where $[ T ]_i$ indicates the instance of the surface code on which $T$ acts.
	Now if $f, f' \in \mathsf{F}_{c'}(\Gamma)$ and $(u,v)\in \mathsf{E}_{c'}(\Gamma)$ such that $u\in f$ and $v\in f'$ and
	$H_{u,v}^{\epsilon_{c'}}$ and $H_{u,v}^{\mu_{c'}}$ are chosen to be independent hopping operators of $f$, then
	\begin{eqnarray}
		\pi(H_{u, v}^{\epsilon_{c'}}) = \left[X_{\tau(u)} X_{\tau(v)}\right]_2 %\label{eq:elec-hopper2}
		\mbox{; } \pi(H_{u , v}^{\mu_{c'}}) = \left[X_{\tau(u)} X_{\tau(v)} \right]_1. \label{eq:mag-hopper2}
	\end{eqnarray}
	
\end{lemma}
\begin{proof}
	We only prove for $H_{u, v}^{\epsilon_c}$ and $H_{u , v}^{\mu_{c'}}$. Similar reasoning can be employed for $H_{u,v}^{\mu_c}$
	and $H_{u, v}^{\epsilon_{c'}}$.
	\begin{compactenum}[(i)]
		\item  $H_{u, v}^{\epsilon_c}$: This operator moves $\epsilon_c$ from $f$ to $f'$ in $\Gamma$. These faces are mapped to adjacent vertices in $\tau(\Gamma)$. By Lemma~\ref{lm:charge-pairing}, $\epsilon_c$ is mapped to $\epsilon_1$, so $\pi(H_{u, v}^{\epsilon})$ should move $\epsilon_1$ from the vertex $\tau(f)$ to the vertex $\tau(f')$ on $\Gamma_{1}$. %There are 
		Many hopping operators can achieve this; choosing the elementary operator gives $\pi(Z_u Z_v)= [Z_{\tau(u,v)}]_1$. Since $\tau(u,v)=\tau(u)=\tau(v)$, Eq.~\eqref{eq:elec-hopper1} follows.
		\item $H_{u, v}^{\mu_{c'}}$: This operator moves $\mu_{c'}$ from $f$ to $f'$. Since $\mu_{c'}$ is mapped to $\mu_1$, $ \pi(H_{u, v}^{\mu_{c'}})$ should move $\mu_1$ from the plaquette $\tau(f)$ to $\tau(f')$ on $\Gamma_{1}$. The operator on the first surface code which achieves this is an $X$-type operator on qubits $\tau(u)$ and $\tau(v)$ in $\Gamma_1$, i.e. $[X_{\tau(u)}X_{\tau(v)}]_1$. 
	\end{compactenum}
	In both cases we choose the hopping operators to be of minimum weight. 
\end{proof}

\begin{figure}
	\centering
	\includegraphics[scale=0.7,angle=0]{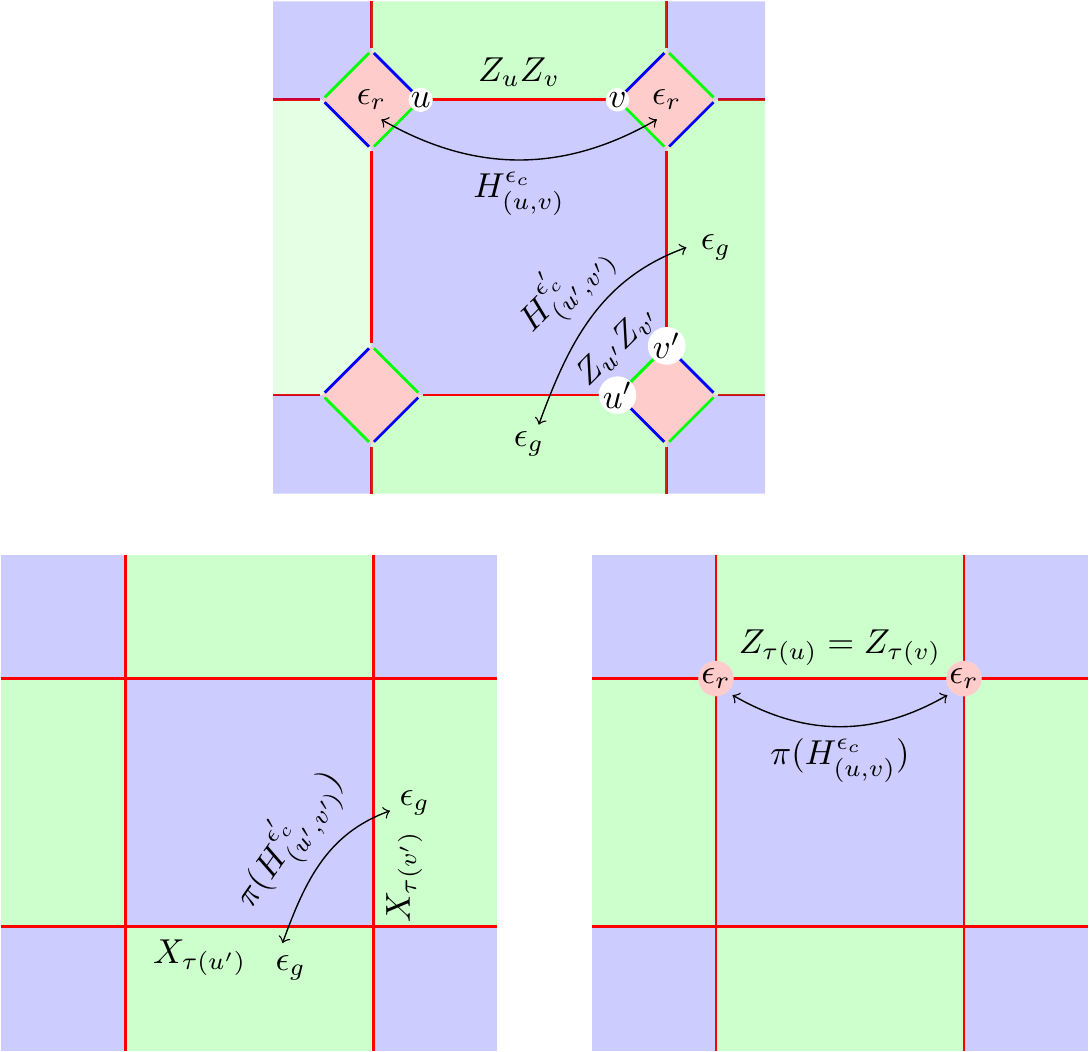}
	\caption{Mapping the independent hopping operators $H_{u,v}^{\epsilon_r}=H_{f_1 \leftrightarrow f_2}^{\epsilon_r}=Z_u Z_v$ and $H_{u',v'}^{\epsilon_g}=H_{f \leftrightarrow f'}^{\epsilon_g}=Z_{u'}Z_{v'}$ on $\Gamma$ onto two copies of $\tau(\Gamma)$ i.e. $\Gamma_1$ and $\Gamma_2$; $\pi(H_{f_1 \leftrightarrow f_2}^{\epsilon_r}) = [Z_{\tau(u)}]_1 $ acts only on $\Gamma_1$ while $H_{f \leftrightarrow f'}^{\epsilon_g} = [X_{\tau(u')}X_{\tau(v')}]_2$ acts only on $\Gamma_2$.} 
	\label{fig:mapping-hops}
\end{figure}

Fig.~\ref{fig:mapping-hops} illustrates how the independent hopping operators of Lemma~\ref{lm:hopping} are mapped. Lemma~\ref{lm:hopping} does not specify the mapping for the dependent hopping operators but it can be obtained as a linear combination of the independent ones. Alternative choices to those given in Lemma~\ref{lm:hopping} exist for $\pi$. These choices are essentially alternate hopping operators on the surface codes which accomplish the same charge movement. 
Such operators can be obtained by adding stabilizer elements to those given in 
Eqs.~\eqref{eq:elec-hopper1}--\eqref{eq:mag-hopper2}.

In this paper we explore the choice when the operators $H_{v_1,v_{2{\ell_f}} }^{\epsilon_{c'}}$ and $H_{v_{2m_f},v_{2m_f+1} }^{\mu_{c'}}$ are among the dependent hopping operators. 
The $c''$-faces form a covering of all the vertices of $\Gamma$ and they are non-overlapping. The elementary hopping operators along the edges on such plaquettes do not interact with the elementary hopping operators of other plaquettes in $\mathsf{F}_{c''}(\Gamma)$. So we can consider each $f\in \mathsf{F}_{c''}(\Gamma)$ independently. This also makes sense from our constraint to keep $\pi$ local. Based on Lemmas~\ref{lm:dep-ops}~and~\ref{lm:hopping}, we can map the independent elementary hopping operators of $f$ along $c$-colored edges. They map elementary hopping operators on $\Gamma$ to elementary hopping operators on $\Gamma_i$.
\begin{align}
	\pi (Z_{v_{2i-1}}Z_{v_{2i}} ) = [Z_{\tau(v_{2i})}]_1 &\mbox{ and } 
	\pi (X_{v_{2i-1}}X_{v_{2i}} ) = [Z_{\tau(v_{2i})}]_2 \label{eq:cedge}
\end{align}

Next, we consider the hopping operators that involve the $c'$-colored edges. Without loss of generality we assume that
$Z_{v_1} Z_{v_{2{\ell_f}}}$  and $X_{v_{2m_f}}X_{v_{2m_f+1}}$ are dependent hopping operators. Then letting $2{\ell_f}+1\equiv 1$ we have
\begin{align}
	\pi(Z_{v_{2i}}Z_{v_{2i+1}}) &= [X_{\tau(v_{2i})}X_{\tau(v_{2i+1})}]_2 \mbox{ ; } 1\leq i < {\ell_f} \label{eq:zz-diag}\\
	\pi(X_{v_{2i}}X_{v_{2i+1}}) &= [X_{\tau(v_{2i})}X_{\tau(v_{2i+1})}]_1 \mbox{ ; } 1\leq i\neq m_f \leq {\ell_f}.
	\label{eq:xx-diag}
\end{align}

All these operators and their images under $\pi$ are linearly independent as can be seen from their supports. From Lemma~\ref{lm:hopping} we obtain the images for the dependent hopping operators:
\begin{align}
	\pi(H_{v_1,v_{2{\ell_f}} }^{\epsilon_{c'}})&= [X_{\tau(v_1)} X_{\tau(v_{2{\ell_f}})}]_2\prod_{i=1}^{\ell_f} [Z_{\tau(v_{2i})}]_1 \\
	\pi(H_{v_{2m_f},v_{2m_f+1} }^{\mu_{c'}}) &= [X_{\tau(v_{2m_f})}X_{\tau(v_{2m_f+1})}]_1 \prod_{i=1}^{\ell_f} [Z_{\tau(v_{2i})}]_2
\end{align}

To complete the map it remains to find the action of $\pi$ for two more independent errors on the color code. One choice is any pair of single qubit operators $X_{v_i}$ and $Z_{v_j}$, where $1\leq i,j\leq 2{\ell_f}$. We can also consider the images under $\pi$. We can see from Eqs.~\eqref{eq:cedge}--\eqref{eq:xx-diag} that the images are also linearly independent and only single qubit $X$-type of errors remain to be generated. One choice is any $[X_{\tau(v_i)}]_1$ on $\Gamma_1$ and $[X_{\tau(v_j)}]_2$ on $\Gamma_2$, where $1\leq i,j\leq 2{\ell_f}$. That is, we need to find $E,E'$ such that $\pi(E)=[X_{\tau(v_i)}]_1$ and $\pi(E')=[X_{\tau(v_j)}]_2$ respect the commutation relations. Lemma~\ref{lm:split} addresses this choice.

\begin{lemma}[Splitting]\label{lm:split} 
	The following choices lead to an invertible $\pi$ while respecting the commutation relations with hopping operators in Eqs.~\eqref{eq:cedge}--\eqref{eq:xx-diag}. 
	\begin{align}
		\pi(g X_{v_1} ) &=[X_{\tau(v_1)}]_1 \mbox{ where } g\in \{I,B_f^X, B_f^Y, B_f^Z \}\label{eq:z-split}\\
		\pi(g Z_{v_{2m_f}} ) &=[X_{\tau(v_{2m_f})}]_2 \mbox{ where } g\in \{I,B_f^X \}\label{eq:x-split}
	\end{align} 
\end{lemma}

\begin{proof}
	Each face $f\in \mathsf{F}_{c''}(\Gamma)$ accounts for $2\ell_f$ qubits, i.e. $4\ell_f$ independent operators. Now, $[X_{\tau(v_1)}]_1$ and $[X_{\tau(v_{2m_f})}]_2$ form a linearly independent set of size $4{\ell_f}$ along with the images of the independent elementary hopping operators on $f$.
	Thus, the elementary hopping operators and the preimages of $[X_{\tau(v_1)}]_1$ and $[X_{\tau(v_{2m_f})}]_2$ account for all the $4\ell_f$ operators on qubits on $f$. Considering all faces in $\mathsf{F}_{c''}(\Gamma)$, we have $\sum_f 4\ell_f=2n$ operators which generate $\mc{P}_{\mathsf{V}(\Gamma)}$. Since their images are independent and $\Gamma_1 \cup \Gamma_2$ has exactly as many qubits as $\Gamma$, $\pi$ must be invertible. 
	
	Next, we prove these choices respect the commutation relations as stated.
	Consider $[X_{\tau(v_1)}]_1$: this error commutes with all the operators in Eq.~\eqref{eq:cedge}--\eqref{eq:xx-diag} except $\pi(Z_{v_1} Z_{v_2}) = [Z_{\tau(v_1)} ]_1$. There are $4\ell_f-3$ such hopping operators on $f$ with which $\pi^{-1}([X_{\tau(v_1)}]_1)$ must commute. As a consequence of the rank-nullity theorem there are $2^{{4\ell_f}-(4\ell_f-3)}$ such operators.
	It can be verified that $\langle X_{v_1}, B_f^X, B_f^Z\rangle $ account for these operators. But $\pi^{-1}([X_{\tau(v_1)}]_1)$ must also anti-commute with $Z_{v_1}Z_{v_2}$. This gives the choices in Eq.~\eqref{eq:z-split} since operators in $\langle B_f^X, B_f^Z\rangle$ commute with $Z_{v_1}Z_{v_2}$. Now let us determine $\pi^{-1}([X_{\tau(v_{2m_f})}]_2)$. Once again with reference to Eq.~\eqref{eq:cedge}--\eqref{eq:xx-diag} we see that it must commute with 
	$4\ell_f-3$ hopping operators on $f$. It also commutes with $\pi^{-1}([X_{\tau(v_1)}]_1)$ since $[X_{\tau(v_{2m_f})}]_2$ commutes with $[X_{\tau(v_1)]_1}$. Again, due to a dimensionality argument there are $2^{4\ell_f-(4\ell_f-2)}$ choices for $\pi^{-1}([X_{\tau(v_{2m_f})}]_2)$. Since 
	$[X_{v_{2m_f}}]_2$ anti-commutes with $[Z_{\tau(v_{2m_f})}]_2$
	its preimage must anti-commute with $\pi^{-1}([Z_{\tau(v_{2m_f})}]_2)=X_{v_{2m_f-1}}X_{v_{2m_f}}$ giving two choices $Z_{v_{2m_f}}$ and $Z_{v_{2m_f}}B_f^X$. We can check that $Z_{v_{2m_f}}$ satisfies all the required commutation relations as does the choice $Z_{v_{2m_f}} B_f^X$. 
\end{proof}

\noindent
In Lemma~\ref{lm:split} we first assigned $\pi^{-1}([X_\tau(v_1)]_1)$ followed by 
$\pi^{-1}([X_{\tau(v_{2m_f})}]_2) $. Changing the order restricts $g$ to $\{I,B_f^X\}$ in Eq.~\eqref{eq:zz-diag} while $g\in \{I,B_f^X, B_f^Y, B_f^Z \}$ in Eq.~\eqref{eq:xx-diag}.

\begin{lemma}[Preserving commutation relations]\label{lm:commutation}
	The map $\pi$ preserves commutation relations of error operators in $\mc{P}_{\mathsf{V}(\Gamma)}$. 
\end{lemma}
\begin{proof}
	We only sketch the proof. It suffices to show that the commutation relations hold for a basis of $\mc{P}_{\mathsf{V}(\Gamma)}$. We consider the basis consisting of the hopping operators along $c$ and $c'$ edges in Eq.~\eqref{eq:cedge}--\eqref{eq:xx-diag} and the single qubit operators given in Lemma~\ref{lm:split}. The proof of Lemma~\ref{lm:split} shows that the commutation relations are satisfied for the single qubit operators. Consider a hopping operator along a $c'$-colored edge. This anti-commutes with exactly two hopping operators along $c$-colored edges on $\Gamma$. For instance, consider $Z_{v_{2i}}Z_{v_{2i+1}}$. From Eq.~\eqref{eq:cedge} --\eqref{eq:xx-diag} this anti-commutes with $X_{v_{2i-1}}X_{v_{2i}}$ and $X_{v_{2i+1}}X_{v_{2i+2}}$. Their images under $\pi$ are $[X_{\tau(v_{2i})}X_{\tau(v_{2i+1})}]_2$, $[Z_{\tau(v_{2i})}]_2$ and $[Z_{\tau(v_{2i+1})}]_2$ for which it is clear that the commutation relations are satisfied. 
	
	The operators along the $c$-colored edges are given in Eq.~\eqref{eq:cedge}. Suppose we consider $\pi(Z_{v_{2i-1}}Z_{v_{2i}})$; then it anti-commutes with $X_{v_{2i-2}}X_{v_{2i-1}}$ and $X_{v_{2i}}X_{v_{2i+1}}$. We only need to verify commutation for those operators which are independent. Assuming that they are both independent, then their images are $X_{\tau(v_{2i-2})}X_{\tau(v_{2i-1})}$ and $X_{\tau(v_{2i})}X_{\tau(v_{2i+1})}$ respectively. They anti-commute with $\pi(Z_{v_{2i-1}}Z_{v_{2i}}) = [Z_{\tau(v_{2i-1})}]_1 = [Z_{\tau(v_{2i})}]_1$.
	If only one of the operators is independent, then we need only verify for that operator. The preceding argument already establishes this result. We can argue in a similar fashion to show that commutation relations are preserved for the operators of the type $X_{v_{2i-1}}X_{v_{2i}}$ and $X_{v_{2i}}X_{v_{2i+1}}$.
\end{proof}
\begin{lemma}[Preserving code capabilities]\label{lm:stabilizers}
	Under $\pi$, stabilizers of the color code on $\Gamma$ are mapped to stabilizers on the surface codes on $\Gamma_1$ and $ \Gamma_2$. 
\end{lemma}
\begin{proof}
	To prove this, it suffices to show that the stabilizers associated with plaquettes of all three colors are mapped to stabilizers on the surface codes. If $\Gamma_i=\tau_c(\Gamma)$, then we show that the stabilizers associated with $f\in \mathsf{F}_{c'}(\Gamma)\cup \mathsf{F}_{c''}(\Gamma)$ are mapped to the plaquette stabilizers on $\Gamma_i$. If $f\in \mathsf{F}_{c'}(\Gamma)$, then $B_f^Z = \prod_{i=1}^{\ell_f} H_{v_{2i-1},v_{2i}}^{\epsilon_{c}}$. By Lemma~\ref{lm:hopping} this is mapped to $\prod_{i}^{\ell_f} [Z_{\tau(v_{2i})}]_1 = \prod_{e\in \partial(\tau(f))}[Z_e]_1$. Using a similar argument we can show that $B_f^Z \in \mathsf{F}_{c''}(\Gamma)$ is also a plaquette stabilizer on $\Gamma_1$. Since faces in $\mathsf{F}_{c'}(\Gamma)\cup \mathsf{F}_{c''}(\Gamma)$ are in one to one correspondence with the faces of $\tau(\Gamma)$, they account for all the face stabilizers on $\Gamma_1$. By considering $B_f^X$, we can similarly show that they map to the face stabilizers on $\Gamma_2$. 
	
	Now consider a face $f\in \mathsf{F}_{c}(\Gamma)$. Consider $B_f^Z$, this can be decomposed into hopping operators $H_{u,v}^{\epsilon_{c'}}$ along $c'$-edges. By Lemma~\ref{lm:hopping}, such an operator maps to $[X_{\tau(u)}X_{\tau(v)}]_2$ and an additional stabilizer on one of the faces of $\Gamma_1$ if $H_{u,v}^{\epsilon_{c'}}$ is a dependent hopping operator. Thus $B_f^Z$ maps to a vertex operator on $\tau(f)$ in $\Gamma_2$ and possibly a combination of plaquette stabilizers. Since every vertex in $\Gamma_2$ is from a face in $\Gamma$, we can account for all the vertex operators on $\Gamma_2$. Similarly, by considering the stabilizer $B_f^X$ we can account for all the vertex operators on $\Gamma_1$.
	\end{proof}
	
	\begin{theorem}\label{th:tcc-map}
		Any 2D color code (on a 2-colex $\Gamma$ without parallel edges) is equivalent to a pair of surface codes $\tau(\Gamma)$ under the map $\pi$ defined as in Algorithm~\ref{alg:tcc-projections}. 
	\end{theorem}
	\begin{proof}[Proof Sketch]
		By charge conservation we require two copies of $\tau(\Gamma)$ to represent the color code using surface codes. Lines 2--3 follow from Lemma~\ref{lm:charge-pairing}. Since $c''$-colored faces in $\mathsf{F}_{c''}(\Gamma)$ cover all the qubits of the color code, we account for all the single qubit operators on the color code by the {$\mathsf {for}$}-loop in lines 4--10. The closed form expressions for single qubit errors in lines 6--9 are a direct consequence of Lemmas~\ref{lm:hopping}, \ref{lm:split} and the choices given in Eqs.~\eqref{eq:cedge}--\eqref{eq:xx-diag} and Eqs.~\eqref{eq:z-split}--\eqref{eq:x-split}. By considering the images of the stabilizers of the color code, we can show that they are mapped to the stabilizers of the surface codes on $\Gamma_i$ (see Lemma~\ref{lm:stabilizers}). From Lemma~\ref{lm:commutation}, the commutation relations among the hopping operators on the color code in Eq.~\eqref{eq:cedge}--\eqref{eq:xx-diag} and the single qubit operators in Eq.~\eqref{eq:z-split}--\eqref{eq:x-split} are preserved. Hence, the errors corrected by the color code are the same as those corrected by the surface codes on $\Gamma_i$. Thus the color code is equivalent to two copies of $\tau(\Gamma)$.
	\end{proof}

	\begin{algorithm}
		% \floatname{algorithm}{Construction}
		\caption{{\ensuremath{\mbox{ Mapping a 2D color code to surface codes}}}}\label{alg:tcc-projections}
		\begin{algorithmic}[1]
			\REQUIRE {A 2-colex $\Gamma$ without parallel edges; $\Gamma$ is assumed to have a 2-cell embedding.}
			\ENSURE { % $\tau: \Gamma \rightarrow (\mathsf{F}_c(\Gamma),\mathsf{E}_c(\Gamma)) $ and 
				$\pi: \mc{P}_{\mathsf{V}(\Gamma)}\rightarrow \mc{P}_{\mathsf{E}(\Gamma_1)}\otimes \mc{P}_{\mathsf{E}(\Gamma_2)}$, where 
				%$\mc{P}_{\mathsf{V}(\Gamma)}$, $\mc{P}_{\mathsf{E}(\Gamma_i)}$ are Pauli groups on vertex and edge sets of $\Gamma$ and 
				$\Gamma_i=\tau_c(\Gamma)$.}
			
			\STATE Pick a color $c\in \{r,g,b\}$ and contract all edges of $\Gamma$ that are colored $\{r,g,b\}\setminus c$ to obtain $\tau(\Gamma)$. Denote two instances of $\tau(\Gamma)$ as $\Gamma_1$ and $\Gamma_2$.
			\STATE Choose charges $\epsilon_{c}$, $\mu_{c}$, $\epsilon_{c'}$ and $\mu_{c'}$ on $\Gamma$, where $c' \neq c$.
			\STATE Set up correspondence between charges on $\Gamma$ and $\Gamma_{i}$ as follows:
			$\epsilon_1 \equiv \epsilon_c $, $\mu_1 \equiv \mu_{c'}$, $\epsilon_2 \equiv \mu_c$
			and $\mu_2 \equiv \epsilon_{c'}$ .
			\FOR { each $c''$-colored face $f$ in $\mathsf{F}(\Gamma)$}
			\STATE Let the boundary of $f$ be $v_1$, \ldots $v_{2{\ell_f}}$.
			\STATE Choose a pair of $c'$-colored edges in $\partial(f)$, say $(v_{2{\ell_f}},v_1)$ and $(v_{2m_f},v_{2m_f+1})$. Let $[ T ]_i$ denote that $T$ acts on $\Gamma_i$. 
			%the instance of the surface code on $T$ acts.
			\begin{align}
				\pi(Z_{v_1}) &= \left[ X_{\tau(v_1)}\right]_2 \prod_{i=1}^{m_f} \left[ Z_{\tau(v_{2i})} \right]_1 %\\
			\end{align}
			\STATE For $1 \leq j\leq {\ell_f} $ compute the mapping (recursively) as 
			\begin{align}
				\pi(Z_{v_{2j}}) &= \pi (Z_{v_{2j-1}}) \left[ Z_{\tau(v_{2j})}\right]_1 \label{eq:zmap_rec_1}\\
				\pi(Z_{v_{2j-1}}) &= \pi (Z_{v_{2j-2}}) \left[X_{\tau(v_{2j-2})} X_{\tau(v_{2j-1})}\right]_2 \label{eq:zmap_rec_2}
			\end{align}
			
			\STATE For $1 \leq j\leq m_f$ compute the mapping as 
			\begin{align}
				\pi(X_{v_1}) &= \left[ X_{\tau(v_1)}\right]_1 \label{eq:xmap_rec01}\\
				\pi(X_{v_{2j}}) &= \pi (X_{v_{2j-1}}) \left[ Z_{\tau(v_{2j})}\right]_2 \label{eq:xmap_rec02}\\
				\pi(X_{v_{2j-1}}) &= \pi (X_{v_{2j-2}}) \left[ X_{\tau(v_{2j-2})} X_{\tau(v_{2j-1})}\right]_1 \label{eq:xmap_rec03}
			\end{align}
			
			\STATE For $m_f+1 \leq j\leq {\ell_f}$ compute the mapping as 
			\begin{align}
				\pi(X_{v_{2{\ell_f}}}) &=  \left[ X_{\tau(v_{2{\ell_f}})}\right]_1\label{eq:xmap-rec04}\\
				\pi(X_{v_{2j-1}}) &= \pi (X_{v_{2j}}) \left[ Z_{\tau(v_{2j})}\right]_2\label{eq:xmap-rec05}\\
				\pi(X_{v_{2j}}) &= \pi (X_{v_{2j+1}}) \left[ X_{\tau(v_{2j})} X_{\tau(v_{2j+1})}\right]_1\label{eq:xmap-rec06}
			\end{align}
			\ENDFOR
		\end{algorithmic}
	\end{algorithm}
	The map given in Algorithm \ref{alg:tcc-projections} is in recursive form. We can get an explicit form of the map as follows,
	For $1 \leq i \leq m_f$
	\begin{align}
		\pi(Z_{v_{2i-1}}) &= [Z_{\tau(v_{2i})}Z_{\tau(v_{2i+2})} \hdots Z_{\tau(v_{2m_f})}]_1 [X_{\tau(v_{2i})}]_2 \\
		\pi(Z_{v_{2i}}) &= [Z_{\tau(v_{2i+2})}Z_{\tau(v_{2i+4})} \hdots Z_{\tau(v_{2m_f})}]_1 [X_{\tau(v_{2i})}]_2 \\
		\pi(X_{v_{2i-1}}) &= [X_{\tau(v_{2i})}]_1 [Z_{\tau(v_2)} Z_{\tau(v_4)} \hdots Z_{\tau(v_{2i-2})}]_2 \\
		\pi(X_{v_{2i}}) &= [X_{\tau(v_{2i})}]_1 [Z_{\tau(v_2)}Z_{\tau(v_4)} \hdots Z_{\tau(v_{2i})}]_2 
	\end{align}
	and for $m_f <i\leq \ell_f$,
	\begin{align}
		\pi(Z_{v_{2i-1}})&=[Z_{\tau(v_{2m_f+2})} Z_{\tau(v_{2m_f+4})} \hdots Z_{\tau(v_{2i-2})}]_1 [X_{\tau(v_{2i})}]_2\\
		\pi(Z_{v_{2i}})&=[Z_{\tau(v_{2m_f+2})} Z_{\tau(v_{2m_f+4})} \hdots Z_{\tau(v_{2i})}]_1 [X_{\tau(v_{2i})}]_2\\ 
		\pi(X_{v_{2i-1}})&=[X_{\tau(v_{2i})}]_1 [Z_{\tau(v_{2i})} Z_{\tau(v_{2i+2})} \hdots Z_{\tau(v_{2\ell_f})}]_2 \\
		\pi(X_{v_{2i}})&=[X_{\tau(v_{2i})}]_1 [Z_{\tau(v_{2i+2})} Z_{\tau(v_{2i+4})} \hdots Z_{\tau(v_{2\ell_f})}]_2
	\end{align}

	We can easily obtain the inverse map $\pi^{-1}$. Using the same notation as in 
	Algorithm~\ref{alg:tcc-projections}, we have the following. 
	\begin{align}
		\pi^{-1}([Z_{\tau(v_{2i-1})}]_1) =\pi^{-1}([Z_{\tau(v_{2i})}]_1)  &=Z_{v_{2i-1}} Z_{v_{2i}} \label{eq:inv_mapz1}\\
		\pi^{-1}([Z_{\tau(v_{2i-1})}]_2) =\pi^{-1}([Z_{\tau(v_{2i})}]_2)  &= X_{v_{2i-1}} X_{v_{2i}} \label{eq:inv_mapz2}
	\end{align}
	For $1 \leq i\leq m_f$ we have 
	\begin{align}
		\pi^{-1}([X_{\tau(v_{2i-1})}]_1) = \pi^{-1}([X_{\tau(v_{2i})}]_1) &=X_{v_{1}} X_{v_{2}} \hdots X_{v_{2i-1}} \label{eq:inv_mapx1}\\
		\pi^{-1}([X_{\tau(v_{2i-1})}]_2) = \pi^{-1}([X_{\tau(v_{2i})}]_2) &=Z_{v_{2i}} Z_{v_{2i+1}} \hdots Z_{v_{2m_f}} \label{eq:inv_mapx2}
	\end{align}
	For $m_f+1 \leq i\leq {\ell_f}$ compute the inverse mapping as 
	\begin{align}
		\pi^{-1}([X_{\tau(v_{2i-1})}]_1) = \pi^{-1}([X_{\tau(v_{2i})}]_1) &=X_{v_{2i}} X_{v_{2i+1}} \hdots X_{v_{2l_f}} \label{eq:inv_mapx3}\\
		\pi^{-1}([X_{\tau(v_{2i-1})}]_2) = \pi^{-1}([X_{\tau(v_{2i})}]_2) &=Z_{v_{2m_f+1}} Z_{v_{2m_f+2}} \hdots Z_{v_{2i-1}} \label{eq:inv_mapx4}
	\end{align}
	A complete example of the mapping for the color code on the square octagon lattice is  illustrated in Fig.~\ref{fig:x-error-map}.
    
	\section{Decoding color codes via surface codes}\label{sec:decoding}
	
	In the previous section we showed that the color code is equivalent to a pair of surface codes. Now, we aim at developing a color code decoder by using this equivalence.
	
	Before we can use this equivalence to decode the color code, we make the following observations. First, observe that the decoding problem on the color code is presented as follows: given the syndrome of an error on the color code, our task is to estimate the error up to a stabilizer. The equivalence we showed in the previous section is due to a (local) map on Pauli operators. However, in the context of decoding we do not have access to the error but only to the syndrome of the error. This motivates the development of a map that projects the syndromes on the color code to syndromes on the surface codes. 
	
	This section is structured as follows. We first project the syndromes on the color code to syndromes on the surface codes.  We then study how the error model on the color code maps to the surface codes. Then, we report the performance of the color code on the square octagonal lattice. Lastly, we use the equivalence to investigate the performance of the color code on the quantum erasure channel. 
	
	\subsection{Projecting the syndromes on color code to surface codes}
	The syndrome for a stabilizer code is obtained by measuring the stabilizer generators. A color code is a CSS code and every face can be associated to an $X$ and $Z$ type stabilizer. Given an error on the color code the syndrome is defined on each face. In projecting the syndrome on the colour code to the surface code we require the following to be satisfied.
	\begin{compactenum}[i)]
		\item Consistency: Suppose $E$ is an error on the color code and $s_E$ its syndrome. We require the syndrome on the surface code to be the syndrome of $\pi(E)$ i.e. it should be syndrome of the image of the error on the surface codes. 
		
		\item Locally computable: Given $s_E$, the syndrome of an error $E$ on the color code, we want to be able to compute the syndrome on the surface codes  only from $s_E$ and not from $E$. Further, we want to be able to compute the syndrome in a local fashion. In other words, the syndrome computation 
		on the surface code must depend only on the syndromes in the ``neighbourhood of the syndrome''. This condition is imposed to ensure that the complexity of syndrome computation is linear. 
	\end{compactenum}

	Consider a stabilizer generator $B_f^\sigma$ on the color code. Let 
	$\pi(B_f^\sigma)= \prod_{i=1}^{n_1} A_{v_i} \prod_{j=1}^{n_2}B_{f_j}$, where 
	$A_{v_i}$ and $B_{f_j}$ are vertex and plaquette stabilizer generators on surface codes. Suppose an error $E$ produces the syndrome $s$ with respect to $B_f^\sigma$. Then we want $\pi(E)$ to produce the same syndrome with respect to $\pi(B_f^\sigma)$. On the surface codes the syndrome is determined by whether $A_{v_i}$ and $B_{f_j}$ commute with $\pi(E)$ or not. 
	
	Recall that the syndrome associated to a stabilizer generator $S$ with respect to an error is given by $SE = (-1)^s E S $, where $s=0$ if $E$ commutes with $S$ and $s=1$ otherwise. Let $B_f^\sigma E = (-1)^s E B_f^\sigma$,  $A_{v_i}E =  (-1)^{s_i}E A_{v_i} $ and $B_{f_j}E =  (-1)^{s_j'}E B_{f_j} $ Therefore, 
	we must have 
	\begin{eqnarray}
		(-1)^s &=& \prod_i  (-1)^{s_i} \prod_j (-1)^{s_j'}\\
		s & =&  s_1+ \cdots +s_{n_1}+  {s_1'}+\cdots s_{n_2}'\label{eq:syndrome-map}
	\end{eqnarray}
	The sum of these syndromes $s_i$ and $s_j'$  must be same as the syndrome for $E$ with respect to $B_f^\sigma$. Here we know $s$, while $s_i$ and $s_j'$ are unknown. From the preceding discussion we see that the syndromes on the surface codes can be found by solving a system of linear equations. This can lead to a super linear overhead in general. Fortunately, we can solve this in linear time since some of the stabilizers $B_f^\sigma$ are mapped to exactly one stabilizer generator on the surface code. In this case the right hand side of Eq.~\eqref{eq:syndrome-map} contains only one variable, in other words 
	the syndrome on the color code is projected directly onto the surface code. 
	This happens whenever $f\in \mathsf{F}_{c'}(\Gamma)\cup \mathsf{F}_{c''}(\Gamma)$. If  $f\in \mathsf{F}_c(\Gamma)$, then as will be shown, the right hand side of Eq.~\eqref{eq:syndrome-map} contains only one unknown variable. Thus we are able to map the syndromes from the color code to the surface codes. The locality of this map follows from the fact that $\pi$ is local so each $B_f^\sigma$ is mapped onto a collection of local stabilizers in each copy of the surface code. We shall now prove this formally in the rest of the section. We shall use the notation $s_f^\sigma$ to denote the syndrome of $B_f^\sigma$. First, we need to find the images of the stabilizers on the color code for which we define a dependent edge set.
	\begin{definition}[Dependent edge set]
		Let $D_X$ be the set of $c'$ colored edges on $\Gamma$ which are defined as the dependent edges for $X$ hopping operators under $\pi$. Likewise $D_Z$ represents the $c'$ colored edges which are dependent for $Z$ hopping operators.
	\end{definition}
	
	Let $f_{c''}^{\ni e}$ be the unique $c''$-face that contains the $c'$-edge $e$ in its boundary. 
	
	\begin{corollary}[Images of stabilizers]\label{co:stab-images}
		Let $f\in \mathsf{F}_{c'}(\Gamma)\cup \mathsf{F}_{c''}(\Gamma)$. Then
		\begin{eqnarray}
			\pi(B_f^Z)&= [B_{\tau(f)}^Z]_1\\
			\pi(B_f^X) &= [B_{\tau(f)}^Z]_2
		\end{eqnarray}
		If $f\in \mathsf{F}_{c}(\Gamma)$, then
		\begin{eqnarray}
			\pi(B_f^X) &= [A_{\tau(f)}^X]_1 \prod\limits_{e\in \partial(f) \cap D_X} [B_{\tau(f_{c''}^{\ni e})}^Z]_2 \label{eq:stab-xc}\\
			\pi(B_f^Z) &= [A_{\tau(f)}^X]_2 \prod\limits_{e\in \partial(f) \cap D_Z} [B_{\tau(f_{c''}^{\ni e})}^Z]_1 \label{eq:stab-zc}
		\end{eqnarray}
		where $D_X$ and $D_Z$ are the set of dependent $c'$-edges for magnetic and electric hopping operators respectively.
	\end{corollary}
	\begin{proof}
		Let $f\in \mathsf{F}_{c'}(\Gamma)\cup \mathsf{F}_{c''}(\Gamma)$.
		Let us number the vertices of $f$ from $1,2,...,2\ell_f$  such that $(1,2)$ is a $c$ colored edge. Let 
		$\tau(v_{2i}) = \tau(v_{2i-1}) = e_i$ where $e_i$ is an edge in $\tau_c(\Gamma)$. 
		Then 
		\begin{align*}
			\pi(B_{f}^X) &= \pi\left( \prod\limits_{i=1}^{\ell_f} X_{v_{2i-1}}X_{v_{2i}} \right) 
			= \prod\limits_{i=1}^{\ell_f} \pi(X_{v_{2i-1}}X_{v_{2i}})  \numberthis \label{eq:img-x-stab}\\
			&\overset{(a)}{=} \prod\limits_{i=1}^{\ell_f} [ Z_{\tau(v_{2i})}]_2 \overset{(b)}{=} \prod\limits_{i=1}^{\ell_f} [ Z_{e_i}]_2 \overset{(d)}{=} [B_{\tau(f)}^{Z}]_2 
		\end{align*}
		where 
		$(a)$ comes from the definition of $\pi$, $(b)$ from the definition of $\tau$ and $(d)$ from the structure of $\tau_c(\Gamma)$. Arguing similarly, we have $\pi(B_f^Z) = [B_{\tau(f)}^{Z}]_1$.

		Now let $f\in \mathsf{F}_c(\Gamma)$. Number the vertices from $1,2,..., 2\ell_f$ such that $(1,2)$ is a $c'$ coloured edge. 
		
		Then,
		\begin{align*}
			\pi(B^X_{f}) &= \pi\left( \prod\limits_{i=1}^{\ell_f} X_{v_{2i-1}}X_{v_{2i}} \right)
			= \prod\limits_{i=1}^{\ell_f} \pi(X_{v_{2i-1}}X_{v_{2i}})
		\end{align*}
		If $e=(v_{2i-1},v_{2i})$ is an independent $c'$-edge, then $\pi(X_{v_{2i-1}}X_{v_{2i}}) = [X_{\tau(v_{2i-1})} X_{\tau(v_{2i})}]_1$, otherwise $\pi(X_{v_{2i-1}}X_{v_{2i}}) = [X_{\tau(v_{2i-1})} X_{\tau(v_{2i})}]_1 [B_{\tau(f_{c''}^{\ni e})}^Z]_2 $. In this case the dependent edge $e$ is in $D_X$, we can write 
		\begin{align*}
			\pi(B^X_{f}) &= \prod\limits_{i=1}^{\ell_f} [X_{\tau(v_{2i-1})}X_{\tau(v_{2i})}]_1 \prod\limits_{e\in \partial(f) \cap D_X} [B_{\tau(f_{c''}^{\ni e})}^Z]_2 \\
			&= \prod\limits_{i=1}^{\ell_f} [X_{\tau(v_{2i-1})}X_{\tau(v_{2i})}]_1 \prod\limits_{e\in \partial(f) \cap D_X } \pi(B^X_{f_{c''}^{\ni e}})\\
			\end{align*}
		Now note that $\{\tau(v_{i})| 1\leq i \leq 2\ell_f\}$ is the set of edges incident on the vertex $\tau(f)$ in $\tau(\Gamma)$.
		This implies that 
		\begin{align*}
			\pi \left(B^X_{f} \right) &= [A_{\pi(f)}]_1 \prod\limits_{e\in \partial(f) \cap D_X } [B_{\tau(f_{c''}^{\ni e})}]_2
		\end{align*}
		We can prove Eq.~\eqref{eq:stab-zc} in a similar fashion. We omit the proof for brevity.
	\end{proof}
	
	Now we show how to compute the syndromes on the surface codes given the the syndromes on color codes.
	
	\begin{theorem}[Face syndromes  on surface codes]\label{th:synd-map-face}
		Let $f\in \mathsf{F}_{c'}(\Gamma)\cup \mathsf{F}_{c''}(\Gamma)$ and 
		$s_{f,i}$ be the syndrome associated with $[B_{\tau(f)}^Z]_i$ of the $i$th surface code.
		Then 
		\begin{eqnarray} 
			s_{f,1}= s_f^{Z} \quad \text{ and } \quad s_{f,2} = s_f^X,\label{eq:face-syn}
		\end{eqnarray} 
		where $s_f^\sigma$ is the syndrome for $B_f^\sigma$ on the color code. 
	\end{theorem}
	\begin{proof}
		
		From Corollary~\ref{co:stab-images} we have $\pi(B_f^X) = [B_{\tau(f)}^{Z}]_2$. By Lemma~\ref{lm:commutation}, the commutation relations of the Pauli operators on the color code are preserved
		under $\pi$. So the syndrome of an error $E$ with respect to $B_f^X$ must be exactly the same as the syndrome of $[B_{\tau(f)}^{Z}]_2$ with respect to $\pi(E)$. 
		
		Similarly, $\pi(B_f^Z) = [B_{\tau(f)}^{Z}]_1$. Therefore, the syndrome of $E$ with respect to $B_f^Z$ is identical to  the syndrome of $[B_{\tau(f)}^{Z}]_1$ with respect to $\pi(E)$. 
		
	\end{proof}
	
	\begin{theorem}[Vertex syndromes on surface codes]\label{th:synd-map-vertex}
		Let $f\in \mathsf{F}_{c}(\Gamma)$ and $s_{f,i}$ be the syndrome associated to $[A_{\tau(f)}^Z]_i$ of the $i$th surface code.
		Then 
		\begin{eqnarray} 
			s_{f,1} &= &s_{f}^X \bigoplus\limits_{e\in \partial(f) \cap D_X } s_{f_{c''}^{\ni e}}^X  \label{eq:vertex-syn-1} \\
			s_{f,2} &= &s_{f}^Z \bigoplus\limits_{e\in \partial(f) \cap D_Z } s_{f_{c''}^{\ni e}}^Z  \label{eq:vertex-syn-2} 
		\end{eqnarray} 
		where $s_f^\sigma$ is the syndrome for $B_f^\sigma$. 
	\end{theorem}
	\begin{proof}
		
		From Eq.~\eqref{eq:stab-zc} we obtain 
		\begin{align*}
			\pi \left(B^X_{f} \prod\limits_{e\in \partial(f) \cap D_X } B^X_{f_{c''}^{\ni e}} \right) &= [A^X_{\pi(f)}]_1\\
		\end{align*}
		From this we immediately obtain the syndrome $s_{f,1}$ as stated. The proof for $s_{f,2}$ is similar and we skip the details. 
	\end{proof}

	We make two important remarks about the induced map on the syndromes. 
	First, note that the syndrome computations are efficient as they only depend on the syndromes of faces adjacent to a given face on the color code. 
	Second, $\pi$ preserves the CSS structure of the color code. 
	A bit flip error causes non-zero syndromes on the vertices of the first copy and faces of the second copy of the surface code and vice versa for a phase flip error. 
	
	We illustrate the syndrome computation with the following example. 
	Figure \ref{fig:synd_map_ex} is an example of syndrome map for bit flip errors. It shows a portion of the square octagonal lattice defining the color code and the associated syndromes, and how they are mapped onto the surface codes. We show a similar computation for the phase flip errors in Fig.~\ref{fig:synd_map_ex_z}.
	\begin{figure}[h]
		\centering
		\includegraphics[scale=0.8]{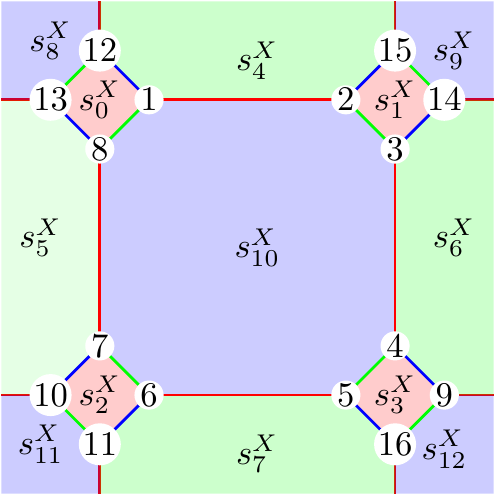} \\
        \vspace{0.2in}
		\includegraphics[scale=1]{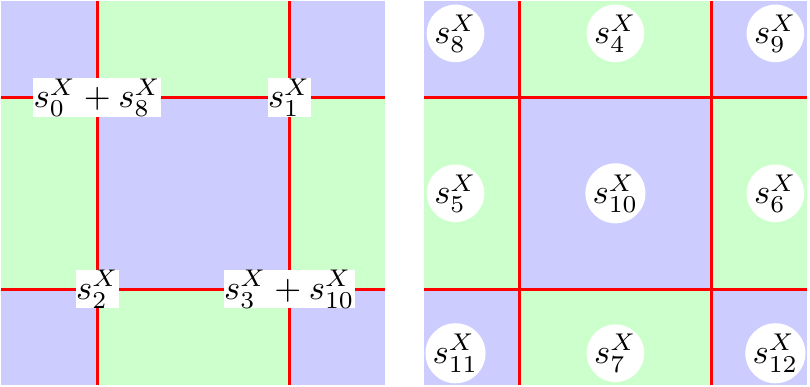}
		\caption{Example of syndrome map for $Z$ errors; we assume that  $X_4X_5$
        and $X_{12}Z_{13}$ are dependent hopping operators.}
		\label{fig:synd_map_ex}
	\end{figure}
	\begin{figure}[h]
		\centering
		\includegraphics[scale=0.8]{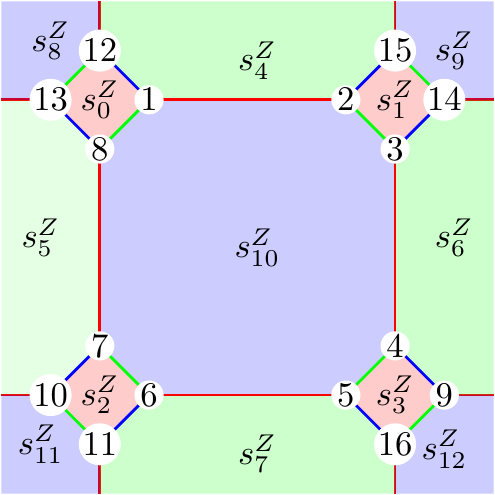} \\
                \vspace{0.2in}
		\includegraphics[scale=1]{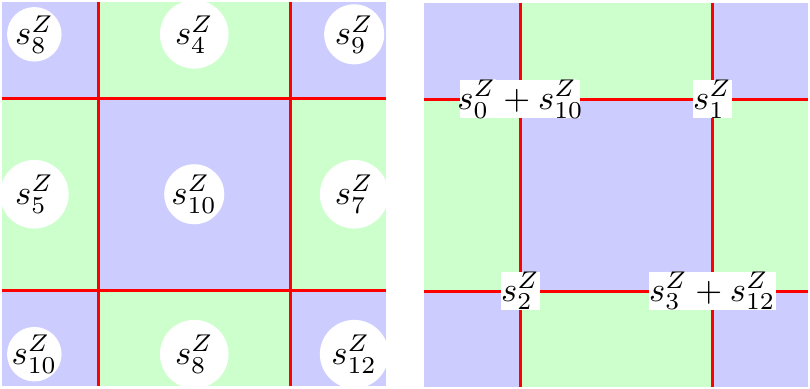}
		\caption{Example of syndrome map for $X$ errors; we assume that  $Z_1Z_8$
        and $Z_{9}Z_{16}$ are dependent hopping operators.}
		\label{fig:synd_map_ex_z}
	\end{figure}
	This example also illustrates how $\pi$ preserves the CSS structure of the color code. Syndromes due to $Z$ errors are mapped to syndromes due to $Z$ type errors only on the first surface code and $X$ errors only on the second surface code. 
	Before we present the simulation results we study the error model on the surface codes. 

	\subsection{Induced error model on surface codes}
	Here, we consider how the error model on the color code translates onto the copies of surface codes under $\pi$. 	
	Let the vertices (equivalently, qubits) on color code be numbered in a sequential order in each $c''$-face, such that $(1,2)$ is a $c'$ edge. %To have a locality nature in circuit, 
	A $c'$-edge $e=(2i-1,2i)$ in the color code also corresponds to the qubits on the surface codes and so far we have indicated the edge on the $i$th copy as $[\tau(e)]_i$. 
	To simplify the notation we shall label the edges in the first copy as $2i-1$ and the edges in the second copy as $2i$.
	\begin{figure}[h]
		\centering
		\includegraphics[scale=0.5]{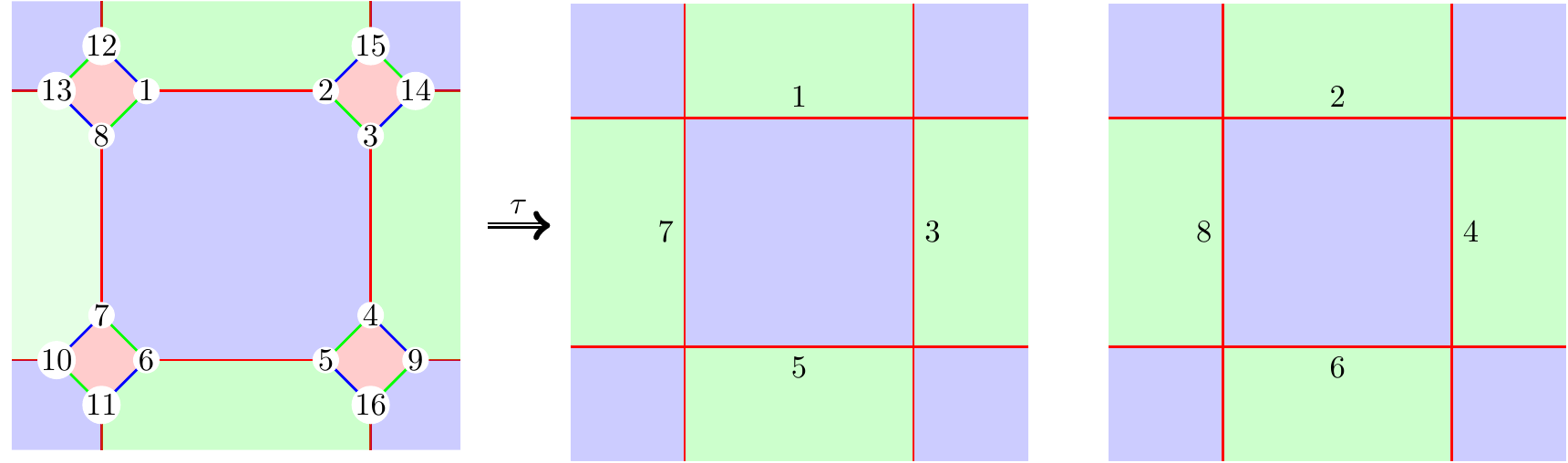}
		\caption{Labeling convention for qubits on the color code and the surface codes}
		\label{fig:num-qubit}
	\end{figure}
	
	Consider the example given in Fig.~\ref{fig:num-qubit}, as per the labeling therein, the single qubit errors are mapped as in Table~\ref{table:so-error-map}.
    \begin{table}[h]
    \centering
\caption{Mapping single qubit errors on square octagon lattice}
\label{table:so-error-map}
\begin{tabular}{c|c||c|c}
\hline
$E$&$\pi(E)$&$E$&$\pi(E)$\\
\hline
$X_{v_1}$&$X_{1}$&$Z_{v_1}$&$X_{2}Z_1Z_3$\\
$X_{v_2}$&$X_{1}Z_{2}$&$Z_{v_2}$&$X_2Z_3$\\
$X_{v_3}$&$X_{3}Z_{2}$&$Z_{v_3}$&$X_{4}Z_3$\\
$X_{v_4}$&$X_{3}Z_{2}Z_4$&$Z_{v_4}$&$X_{4}$\\
$X_{v_5}$&$X_{5}Z_6Z_8$&$Z_{v_5}$&$X_{6}$\\
$X_{v_6}$&$X_{5}Z_8$&$Z_{v_6}$&$X_{6}Z_5$\\
$X_{v_7}$&$X_{7}Z_8$&$Z_{v_7}$&$X_{8}Z_5$\\
$X_{v_8}$&$X_{7}$&$Z_{v_8}$&$X_{8}Z_{5}Z_7$\\
\hline
\end{tabular}
\end{table}
From Table~\ref{table:so-error-map} we infer that even if the error model on the color code is identical and independent on each qubit, the induced error model on the surface codes
	is not. 
	For the purposes of decoding, we should account for the correlations between the two copies of surface code
	as well the correlations within each copy. 
	If we plan to decode the two copies independently, then the error models marginalized on each copy are of relevance. 
    To first order we consider the marginal error probabilities on each copy, although the errors on the two copies of surface code are correlated.

	From Table.~\ref{table:so-error-map}, we can observe that for a $X_{1}$ error  to occur on the surface codes  either of $X_{v_1}$ or $X_{v_2}$ must occur on the color code but not both. If $p_1$ and $p_2$ are the probabilities of error $X_{v_1}$ and $X_{v_2}$ respectively then, the probability of error $X_{1}$ is $p_1(1-p_2)+(1-p_1)p_2$. In a similar fashion, we can find the marginal probability of error for each of the qubits on the surface codes.
	These are given in Table~\ref{table:so-error-prob}. In addition to the marginal probabilities for single qubit errors, we have  also have given expressions for multi-qubit errors on each copy.

    \begin{table}%[h]
    \centering
    \caption{Error probabilities on each surface code}
\label{table:so-error-prob}
\begin{tabular}{c|c|}
\hline
Error on surface code & Probability of error\\
\hline
$Z_1$ & $q$\\
$Z_3$ & $3q(1-q)^2+q^3$\\
$Z_5$ & $3q(1-q)^2+q^3$\\
$Z_7$ & $q$\\ 
\hline
$Z_1Z_3$ & $q(1-q)^2+q^3$\\
$Z_5Z_7$ & $q(1-q)^2+q^3$\\
$Z_1Z_5$ & $q(3q(1-q)^2+q^3)$ \\
$Z_1Z_7$ & $q^2$\\ 
$Z_3Z_5$ & $(3q(1-q)^2+q^3)^2$\\
$Z_3Z_7$ & $q(3q(1-q)^2+q^3)$\\ \hline
$Z_1Z_3Z_5$ & $(3q(1-q)^2+q^3)(q(1-q)^2+q^3)$\\
$Z_1Z_3Z_7$ & $q(q(1-q)^2+q^3)$\\
$Z_1Z_5Z_7$ &$q(q(1-q)^2+q^3)$\\
$Z_3Z_5Z_7$ & $(3q(1-q)^2+q^3)(q(1-q)^2+q^3)$\\\hline
$Z_1Z_3Z_5Z_7$ & $(q(1-q)^2+q^3)^2$\\\hline
$X_2,X_4,X_6,X_8$& $2q(1-q)$\\
\hline
\end{tabular}
\begin{tabular}{|c|c}
\hline
Error on surface code & Probability of error\\
\hline
$Z_2$ & $3p(1-p)^2+p^3$\\
$Z_4$ & $p$\\
$Z_6$ & $p$\\
$Z_8$ & $3p(1-p)^2+p^3$\\\hline
$Z_2Z_4$ & $p(1-p)^2+p^3$\\
$Z_6Z_8$ & $p(1-p)^2+p^3$\\
$Z_2Z_6$ & $p(3p(1-p)^2+p^3)$ \\
$Z_2Z_8$ & $(3p(1-p)^2+p^3)^2$\\
$Z_4Z_6$ & $p^2$\\
$Z_4Z_8$ & $p(3p(1-p)^2+p^3)$\\\hline
$Z_2Z_4Z_6$ & $p(p(1-p)^2+p^3)$\\
$Z_2Z_4Z_8$ & $(p(1-p)^2+p^3)(3p(1-p)^2+p^3)$\\
$Z_2Z_6Z_8$ &$(p(1-p)^2+p^3)(3p(1-p)^2+p^3)$\\
$Z_4Z_6Z_8$ & $p(p(1-p)^2+p^3)$\\ \hline
$Z_2Z_4Z_6Z_8$ & $(p(1-p)^2+p^3)^2$\\ \hline
$X_1,X_3,X_5,X_7$& $2p(1-p)$\\
\hline
\end{tabular}
\end{table}

	We assume a bit flip channel with uniform error probability on the color code, with the probability of a bit flip error being $p$. Let the $X$ and $Z$ error probability on edge $e$ of the $i$th toric code be $\tilde{p}_e^{(i)}$ and $\tilde{q}_e^{(i)}$ respectively. Then we obtain the following (marginal) 
	error probabilities on the surface code due to $\pi$. 
	\begin{subequations}
		\begin{eqnarray*}
			\tilde{p}_1^{(1)} = \tilde{p}_3^{(1)} =\tilde{p}_5^{(1)} = \tilde{p}_7^{(1)}&=& 2p (1-p)\\
			\tilde{q}_2^{(2)} = \tilde{q}_8^{(2)}&=& 3p (1-p)^2+ p^3 \\
			\tilde{q}_4^{(2)} = \tilde{q}_6^{(2)}&=& p
		\end{eqnarray*}
	\end{subequations}
	Note that that bit flip errors induce bit flip errors on one copy of the surface codes and phase flip error errors on the other copy. 
	We generalize the above result to find an expression for the induced error probability on the surface code for arbitrary color code lattice. 
	\begin{theorem}[ Single qubit  marginal  error probabilities under bit flip channel]
		\label{lm:probability_bound} 
		Let $p$ be the probability of $X$ error on color code, and $\tilde{p}_j^{(i)},\tilde{q}_j^{(i)}$ be the marginal probability of $X$ and $Z$ errors on edge $j$ of $i^{th}$ surface code. % that rise because of the map. 
		Then for each $c''$-face $f$, %\ps{Is notation consistent with the previous sections and figure?}
		\begin{align}
			\tilde{p}_{2j-1}^{(1)} &= 2p(1-p) ; 1\leq j \leq \ell_f \label{eq:err-prob-x}\\
			\tilde{q}_{2j}^{(2)} &=\sum\limits_{i=1}^{m_f-j+1} {2m_f-2j+1\choose{2i-1}} p^{2i-1} (1-p)^{2m_f-2j+2-2i} ;\nonumber \\ & \hspace{100pt} 1\leq j\leq m_f \label{eq:err-prob-z1}\\
			\tilde{q}_{2j}^{(2)} &=\sum\limits_{i=1}^{j-m_f} {2j-2m_f-1 \choose{2i-1}} p^{2i-1} (1-p)^{2j-2m_f-2i} ;\nonumber\\ & \hspace{100pt} m_f<j\leq \ell_f \label{eq:err-prob-z2}
		\end{align}
	\end{theorem}
	\begin{proof}
		From Algorithm~\ref{alg:tcc-projections}, we see that an $X$ error on qubit $e_{2j-1}$ happens if and only if there is a $X$ error on qubits $v_{2j-1}$ or $v_{2j}$. If both qubits $v_{2j-1}$ and $v_{2j}$ carry an $X$ error , then the effect gets cancelled out in $e_{2j-1}$ and hence no error will be seen. With this condition, probability of having an $X$ error in qubit $e_{2j-1}$ in the first toric code is equal to probability of having an $X$ error in any one of the qubits $v_{2j-1}$ and $v_{2j}$. This probability is $2p(1-p)$ which is given in Eq.~\eqref{eq:err-prob-x}.
		
		Now consider a qubit $e_{2j}$ on the second toric code. A phase flip error occurs in this qubit if any one of the qubits $v_{2j},v_{2j+1} ,\hdots v_{2m_f}$ (a total of $2m_f-2j+1$ qubits) carries a single qubit $X$ error. Any even combination of them will cancel out while the odd combination survives. Therefore, the probability of error $Z_{e_{2j}}$ is as follows
		
		\begin{align*}
			\tilde{q}_{2j}^{(2)} &=\sum\limits_{i=1}^{m_f-j+1} {2m_f-2j+1\choose{2i-1}} p^{2i-1} (1-p)^{2m_f-2j+2-2i}
		\end{align*}
		
		In a similar way, if $j >m_f$, then $Z_{e_{2j}}$ occurs if there are an odd number of $X$ errors in $v_{2m_f+1},v_{2m_f+2} ,\hdots, v_{2j-1}$. From this we obtain Eq.~\eqref{eq:err-prob-z2}.
	\end{proof}
	
	We now make a few observations regarding the induced error model. First, as already pointed out, the error model on the color code is identical but the error model on the surface codes is not identical. Second, the error model is independent on the color code while errors on the surface codes are correlated. Third, while we have only bit flip errors on the color code, we have bit flip errors on one of the copies of surface code while the other copy has only phase flip errors. 	Finally, we can also derive the following bounds on the marginal error probabilities on the surface codes. 
	\begin{align*}
		\tilde{p}^{(1)} &= 2p(1-p) \\
		p &\leq \tilde{q}^{(2)} \leq \sum\limits_{i=1}^{m^*} {2m^*-1\choose{2i-1}} p^{2i-1} (1-p)^{2m^*-2i}
	\end{align*}
	where 
	\begin{eqnarray}
		m^* = \max_{f\in \mathsf{F}_{c''}(\Gamma)} \{m_f, \ell_f-m_f\} \label{eq:mstar}
	\end{eqnarray}
	
	Assuming phase flip errors on the color code we have the following error model for the surface codes. The proof is similar to Theorem~\ref{lm:probability_bound}, we omit the details. 
	
	\begin{theorem}[Single qubit  marginal  error probabilities  under phase flip channel]
		\label{lm:induced-Z-errors}
		Let $q$ be the probability of $Z$ error on color code, and $\tilde{p}_j^{(i)},\tilde{q}_j^{(i)}$ be the marginal probabilities of $X$ and $Z$ errors on edge $j$ of $i^{th}$ toric code. Then for each $c''$-face $f$, 
		\begin{align*}
			\tilde{p}_{2j}^{(2)} &= 2q(1-q) ; 1\leq j \leq \ell_f \\
			\tilde{q}_{2j-1}^{(1)} &=\sum\limits_{i=1}^{j+1} {2j-1\choose{2i-1}} q^{2i-1} (1-q)^{2j-2i} ; 1\leq j\leq m_f\\
			\tilde{q}_{2j-1}^{(1)} &=\sum\limits_{i=1}^{\ell_f-j+1} {2\ell_f-2j+1 \choose{2i-1}} q^{2i-1} (1-q)^{2\ell_f-2j-2i+2} ;\\ & \hspace{100pt} m_f<j\leq \ell_f\\
		\end{align*}
	\end{theorem}

	\subsection{Performance over the bit flip channel}
	
	We now report the performance of a hard decision decoder for the color codes on the square octagonal lattice. 
	Our decoder is based on the equivalence presented in Section~\ref{sec:map}. 
	The color code is decoded via surface codes. Algorithm~\ref{alg:decoder} contains the details.

	\begin{algorithm}[h]
		\caption{{\ensuremath{\mbox{ Decoding 2D color code via surface codes}}}}\label{alg:decoder}
		\begin{algorithmic}[1]
			\REQUIRE {A 2-colex $\Gamma$ and the syndrome on $\Gamma$.}
			\ENSURE { Error estimate $\hat{E}$}

			%\STATE For $m+1 \leq j\leq {\ell_f}$ compute the mapping as 
			\STATE Project the syndromes obtained on the color code using Eq.~\eqref{eq:face-syn}--%\ref{eq:vertex-syn-1} and
			\eqref{eq:vertex-syn-2}. 
			
			\STATE Decode both the surface codes and obtain the error estimate on toric code $\hat{E}_s$.
			
			\STATE Lift $\hat{E}_s$ to color code using Eqs.~\eqref{eq:inv_mapz1}--\eqref{eq:inv_mapx4}. 
			$$\hat{E} = \pi^{-1}(\hat{E}_s)$$
		\end{algorithmic}
	\end{algorithm}

	The syndromes are measured in the color code lattice and mapped to toric code using theorem \ref{th:synd-map-face} and \ref{th:synd-map-vertex}. Next using this syndrome information on toric code, toric code decoder is used to detect the error estimate on toric cdoes. The errors on the surface code are decoded using minimum weight perfect matching algorithm
	\cite{dennis01}. 
	Finally this error estimate is lifted back to color code using 
	Eqs.~\eqref{eq:inv_mapz1}-\eqref{eq:inv_mapx4}. 
	
	Simulation results for the bit flip channel are shown in Fig.~\ref{fig:sim_4}. For this simulation we assumed 
	$\mathsf{F}_{c''}(\Gamma)$ to be $c''$-colored octagons   and $m_f=2$ for all $ f \in \mathsf{F}_{c''}(\Gamma)$. 
	\begin{figure}[htb]
			\centering
		\includegraphics[scale=0.85]{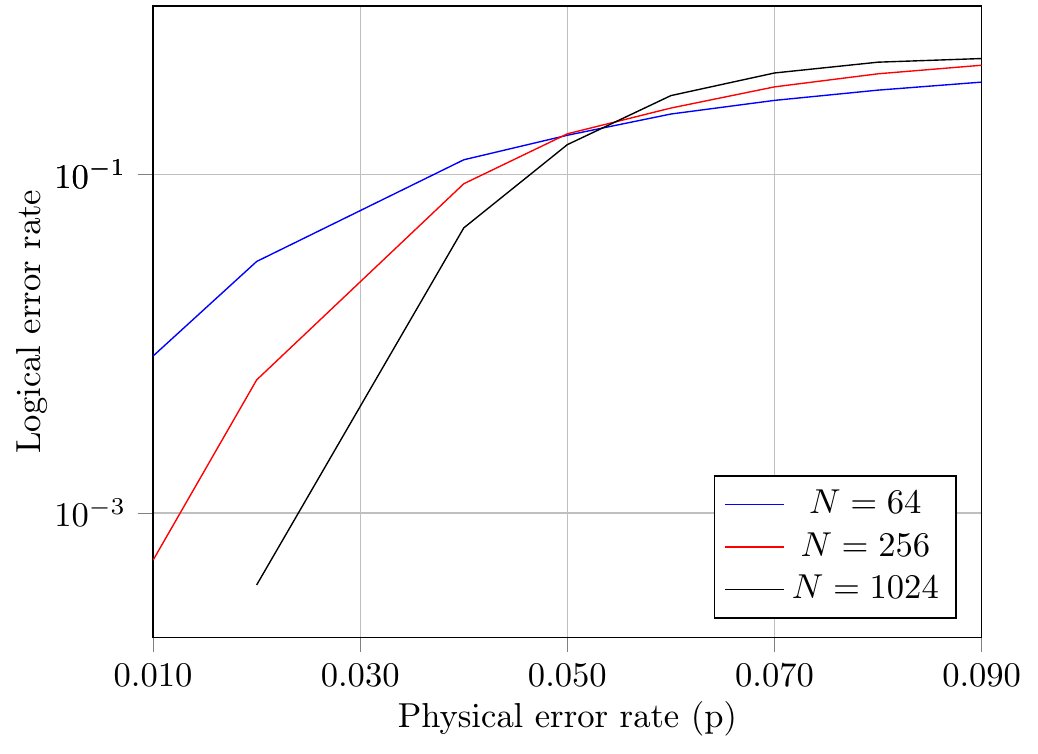}
		\caption{Simulation results for  color code on square octagon lattice for the bit flip channel.}
		\label{fig:sim_4}
	\end{figure}
	We obtained a threshold of approximately $5.3\%$. 
	Note that the perfect matching decoder on the surface code assumes independent and uniform error model.
	However, this will be suboptimal in the present case because the noise on the surface code is correlated and not uniform.
	An improvement can be obtained by modifying the matching decoder to account for this nonuniform error model. 
	Alternatively, we can  incorporate soft decision decoding as in the case of \cite{bombin12}.

	We made an interesting observation during our simulations. The threshold does not seem to vary significantly with respect which $c'$-edges are chosen to be dependent.

\subsection{Performance over the quantum erasure channel}
	In this section we study the performance of color codes over the quantum erasure channel. As in the case of the bit flip channel, we shall apply the equivalence between color codes and copies of surface codes for decoding color codes over the erasure channel. 
	
	We model the erasure channel as follows  (see also \cite{delfosse17}). We replace each of the erased qubits in the color code by a qubit in the completely mixed state $I/2$. Then, we perform the syndrome measurement. Since 
	$I/2= \frac{\rho+X \rho X  + Z \rho Z  + Y\rho Y }{4}$, 
	for any single qubit density matrix $\rho$. We can project this into any of single qubit Pauli errors  where each error occurs  with equal probability. We can then measure the syndrome on the color code and map it onto the surface codes. Finally, we estimate the error using the projected syndrome on the surface codes and lift the error back to the color code. 
	
	A subtle point must be borne in mind when dealing with the erasures. While the locations of erasures on the color code can be explicitly specified, the locations of the erasures on the copies of the induced surface code are not obvious. An erasure on the color code can affect multiple qubits of the surface codes.
	
	A {naive method} to map the erasure locations onto the surface codes is as follows. Suppose the $i^{\text{th}}$ qubit is erased. Then, on the surface codes erase those qubits which are in the support of $\pi(X_i)$, $\pi(Z_i) $ or $\pi(Y_i)$. 
	However, observe that $\supp(\pi(Y_i)) = \supp(\pi(X_i))\cup \supp(\pi(Z_i))$ and hence for any Pauli error $P_{\mathcal{E}}$ acting non trivially on the set of qubits $\mathcal{E}$,  $\supp(\pi(P_{\mathcal{E}})) \subseteq \cup_{i\in \mathcal{E}} \supp(\pi(Y_i)) $. 
	Suppose $\mathcal{E}$ is the erasure pattern and 
	$\pi (\mathcal{E}) $ the set of qubits on which the erasure $\mathcal{E}$ is mapped. Then, we can extend the map $\pi$ from color codes to surface codes to erasures as follows: 
	\begin{eqnarray}
		\pi(\mathcal{E}) = \cup_{i\in \mathcal{E}} \supp(\pi(Y_i)) .\label{eq:erasure-map-simple}
	\end{eqnarray}
	\begin{figure}[t]
		\centering
		\includegraphics[scale=0.5]{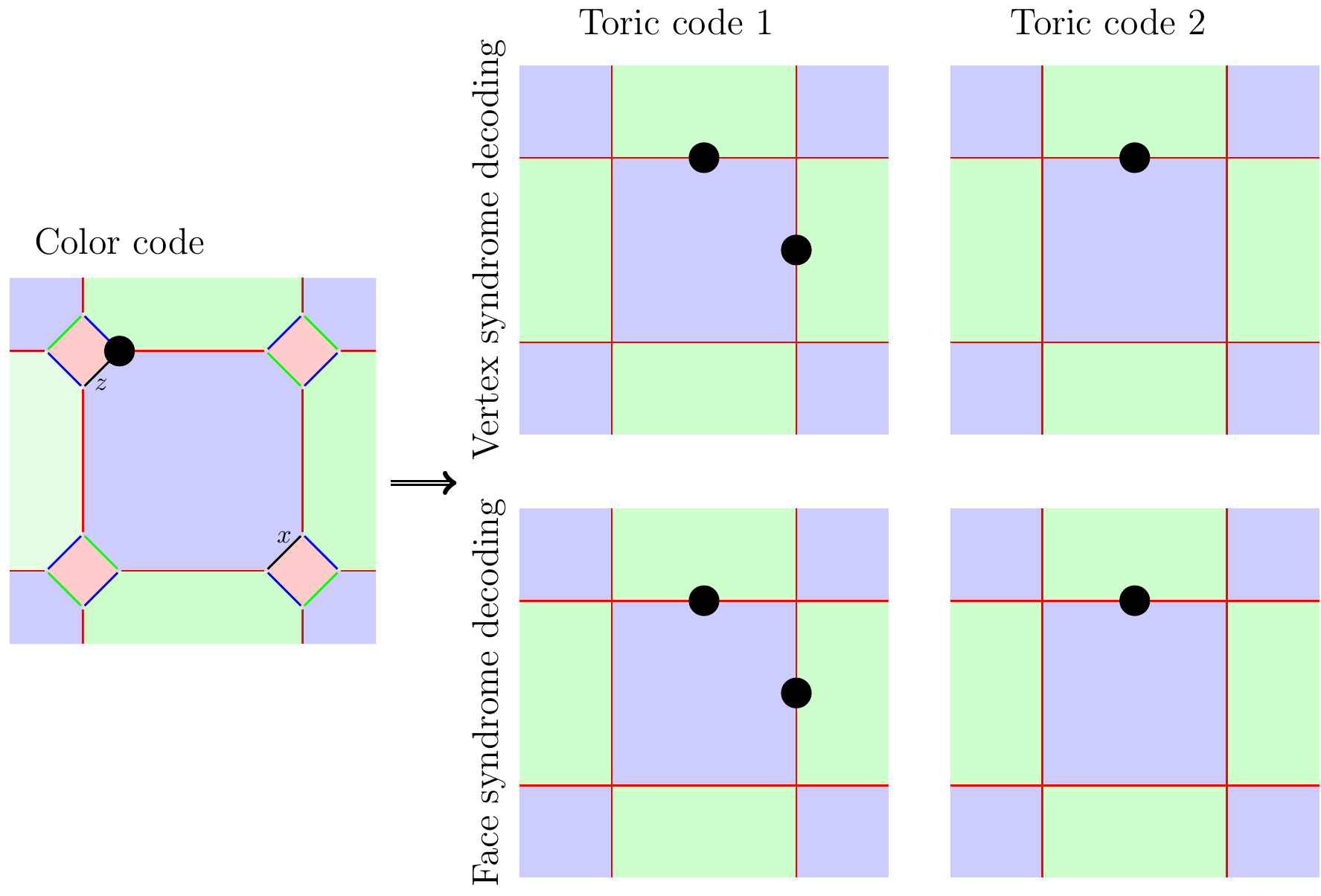}
		\caption{Map of a single erasure on color code based on the  naive method given in Eq.~\eqref{eq:erasure-map-simple}.}
		\label{fig:naiv-erasure-map}
	\end{figure}
	
	While this map is conceptually simple,  it is pessimistic placing more erasures on the surface codes than required. 
	Certain combinations of erasures lead to a smaller set of erasures on the surface codes than this map requires.  
	A decoder based on such an erasure mapping does not perform well. We observed a threshold of about {$21\%$} with such an approach.
	
	{
		Suppose a qubit is erased.  Then we replace this qubit with completely mixed state and measure stabilizer generators. As we discussed earlier, this measurement 
		include one of the Pauli errors on the qubit with  equal probabilities. Our objective is to estimate this error from the measured syndrome and position of erasures.
		Since we are dealing with a CSS code and two copies of surface codes, both bit flip errors  and phase flip errors lead 
		to two instances of decoding each. 
		In the naive map, we used same erasure positions in each of these instances.
		See Fig.~\ref{fig:naiv-erasure-map} for an illustration of this when one qubit is erased. 
		
		Suppose that a $X$  error was induced by the stabilizer measurement, then this error will cause a nonzero syndrome on only two instances of decoding.
		So it would be unnecessary to place erasures on the remaining instances. 
		Similarly if a $Z$ error is induced, then it will be decoded using the remaining two instances and there is no need to place the erasures on those instances which are used for bit flip decoding. 
		}
	
	In other words, we map the erasures to each of the instances consistent with the image of error on that decoding instance. 
	For the decoding instances corresponding to bit flip errors, we have 
	
	\begin{eqnarray}
		\pi(\mathcal{E}) = \cup_{i \in \mathcal{E}} \supp (\pi (X_i) ).\label{eq:erasure-map-x}
	\end{eqnarray}
	For the decoding instances corresponding to the phase flip errors, 
	we have 
	\begin{eqnarray}
		\pi(\mathcal{E}) =  \cup_{i \in \mathcal{E}} \supp (\pi (Z_i) ).\label{eq:erasure-map-z}
	\end{eqnarray}
	
	\begin{figure}[t]
		\centering
		\includegraphics[scale=0.5]{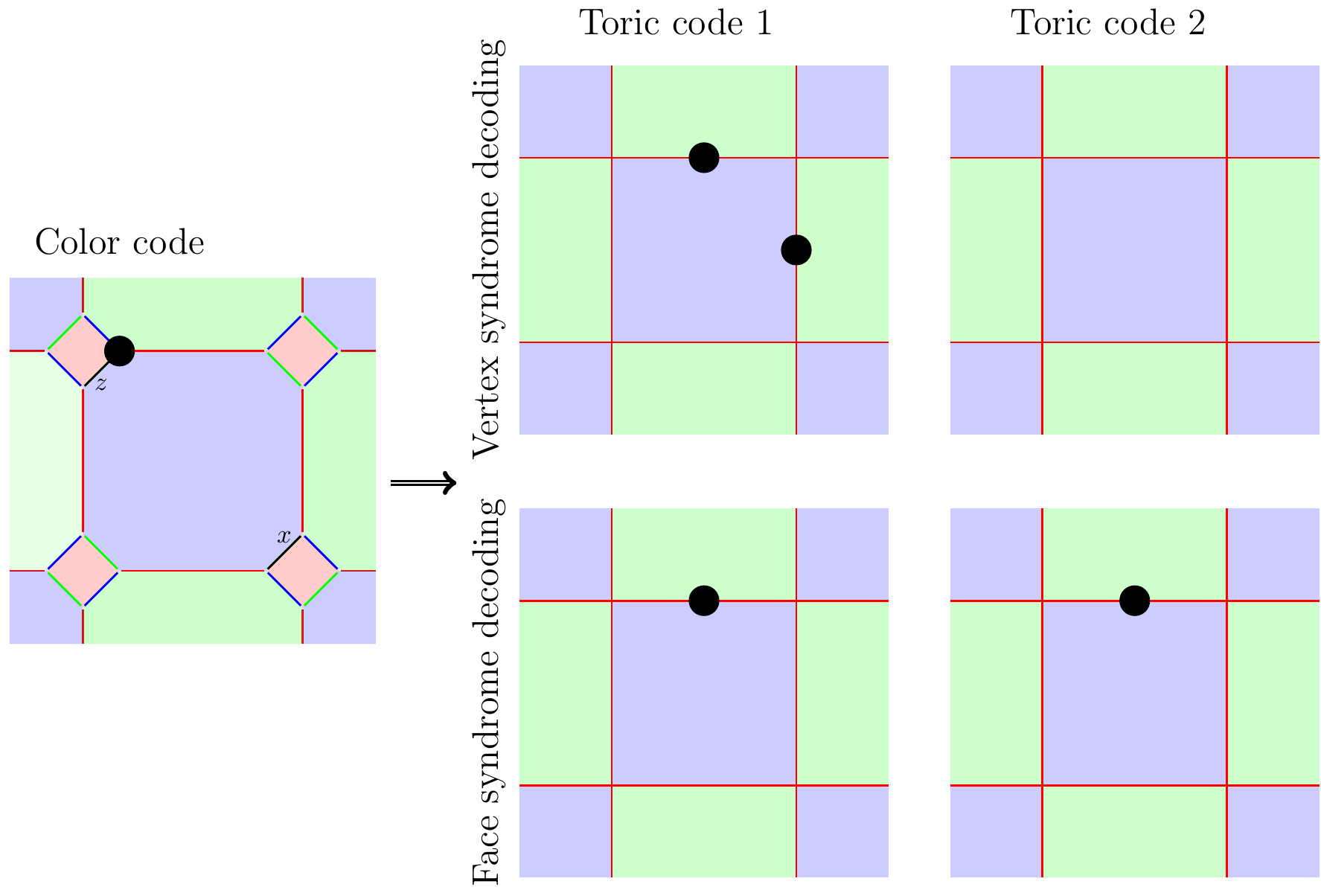}
		\caption{Map of a single erasure on color code based on the improved method given in Eqs.~\eqref{eq:erasure-map-x}-\eqref{eq:erasure-map-z}.}
		\label{fig:simple-erasure-map}
	\end{figure}

	\begin{algorithm}
		\caption{{\ensuremath{\mbox{ Erasure decoding of 2D color codes }}}}\label{alg:erasure-decoder}
		\begin{algorithmic}[1]
			\REQUIRE {A 2-colex $\Gamma$ and an erasure pattern $\mathcal{E}$.}
			\ENSURE { Error estimate $\hat{E}$}
			
			\STATE Replace each erased qubit by a qubit in the completely mixed state. 
			\STATE Perform syndrome measurement on the color code.
			\STATE Map the syndromes obtained on the color code using Eq.~\eqref{eq:face-syn}--%\ref{eq:vertex-syn-1} and
			\eqref{eq:vertex-syn-2}. 
			\STATE // Map erasures onto the surface codes as follows. 
			\STATE For face syndromes: If there an erasure on qubit $v_{2i-1}$ or $v_{2i}$ place an erasure on qubit $\tau(v_{2i-1})$
			on both copies of the surface codes.
			\STATE For vertex syndromes: 
			\STATE //Map erasures onto the second copy of surface code
			\IF{$1\leq i\leq 2m_f$ }
			\IF{there is an erasure among $v_{2i},v_{2i+1}, \ldots, v_{2m_f}$,} place an erasure on $\tau(v_{2i-1})$ \ENDIF
			\ELSE
			\IF{there is an erasure among $v_{2m_f+1},v_{2m_f+2}, \ldots v_{2i-1}$,} place an erasure on $\tau(v_{2i-1})$ \ENDIF 
			\ENDIF 
			\STATE //Map erasures onto the first copy of surface code
			\IF{$1\leq i\leq 2m_f$ }
			\IF{there is an erasure among $v_1,v_2, \hdots v_{2i-1}$,}  place an erasure on $\tau(v_{2i-1})$ \ENDIF
			\ELSE
			\IF{there is an erasure among $v_{2i},v_{2i+1}, \hdots v_{2l_f}$,} place an erasure on $\tau(v_{2i-1})$\ENDIF
			\ENDIF 
			\STATE Decode both the surface codes and obtain the error estimate on surface codes $\hat{E}_s$.
			
			\STATE Lift $\hat{E}_s$ to color code using Eqs.~\eqref{eq:inv_mapz1}--\eqref{eq:inv_mapx4}. 
			$$\hat{E} = \pi^{-1}(\hat{E}_s)$$
		\end{algorithmic}
	\end{algorithm}
	
	This method is described in Algorithm~\ref{alg:erasure-decoder} and gives a threshold of {$30.8\%$}. Simulation results are shown in Fig.~\ref{fig:erase-thresh}.  Just before the submission we came to know of a decoder for the color codes which gives a threshold of 46\% \cite{vidola18}. This suggests
 that the performance of the method discussed here can be improved. One approach to improve the performance of our decoders is to jointly decoding the two copies of the surface codes.

	\begin{figure}
			\centering
		\includegraphics[scale=0.85]{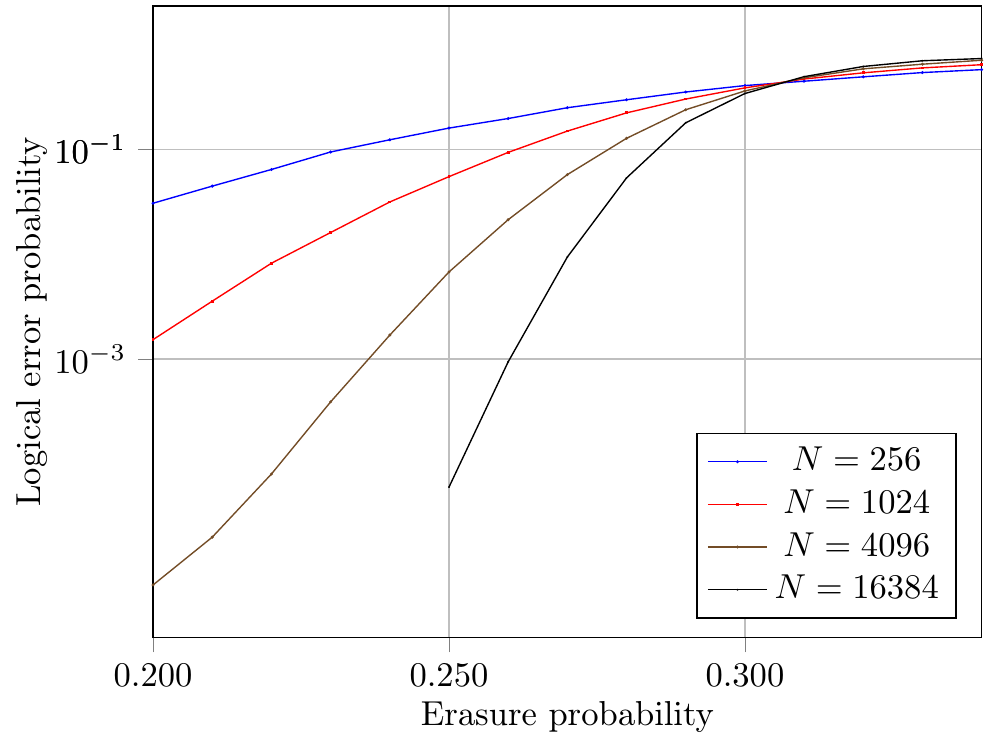}
		\caption{Threshold for quantum erasure channel of color code using the erasure map in Eq.~\eqref{eq:erasure-map-x}--\eqref{eq:erasure-map-z}.}
		\label{fig:erase-thresh}
	\end{figure}
	
	\section{Circuits for Code Transformation} \label{sec:ckts}
	In this section, we give efficient local Clifford circuits for mapping 2D color codes to surface codes. These circuits will be useful in the context of fault tolerant quantum computation. These circuits use only CNOT, H and swap gates. 
	We shall simplify notation slightly in this section for clarity of presentation. 
	\begin{remark}[Notation]\label{rm:notation}
		We continue to assume the color code is defined on a 2-colex $\Gamma$ and the related surface codes are defined on $\tau_c(\Gamma)$. For each face $f\in \mathsf{F}_{c''}(\Gamma)$ we assume that there are $2\ell_f$ vertices (qubits) and number the qubits as $\{1,\ldots 2\ell_f \}$. The dependent $X$ and $Z$ hopping operator are associated to the edges $(1,2\ell_f)$ and $(2m_f, 2m_f+1)$ respectively. 
		Recall that we used the notation $[ T]_i$ to indicate $T$ acts on the $i$th copy. Since we use the odd labels for the qubits on the first surface code and even labels for the second copy we shall drop the brackets and subscript. 
	\end{remark}
	
	With this revised notation, we obtain the following expressions for the map. For $1 \leq i \leq m_f$, we have
	\begin{align}
		\pi(Z_{{2i-1}}) &= X_{2i}\prod_{j=i}^{m_f} Z_{2j-1} \label{eq:zmap-1-a}\\
		\pi(Z_{{2i}}) 
		&=X_{2i} \prod_{j=i+1}^{m_f} Z_{2j-1} \label{eq:zmap-1-b}%\\
	\end{align}
	and for $m_f <i\leq \ell_f$
	\begin{align}
		\pi(Z_{{2i-1}}) &= X_{2i} \prod_{j=m_f+1}^{i-1} Z_{2j-1}  \label{eq:zmap-2-a}\\
		\pi(Z_{{2i}}) &= X_{2i} \prod_{j=m_f+1}^{i} Z_{2j-1} \label{eq:zmap-2-b} 
	\end{align}
	To transform $Z_i$ as defined by Eqs.~\eqref{eq:zmap-1-a}--\eqref{eq:zmap-2-b} we introduce the gates $U_{2i}$ and $V_{2i+1}$:
	\begin{eqnarray}
		\text{U}_{2i}& = &\text{CX}_{1}^{2i-1}\text{CX}_{2}^{2i-1} \cdots \text{CX}_{2i-2}^{2i-1}\text{CX}_{2i-1}^{2i} 
		\label{eq:u2i}\\
		V_{2i+1}& = &\text{CX}_{2l_f}^{2i+2}\text{CX}_{2l_f-1}^{2i+2} \hdots \text{CX}_{2i+3}^{2i+2}\text{CX}_{2i+2}^{2i+1} 
		\label{eq:v2i}
	\end{eqnarray}
	where $\text{CX}_t^c$ indicates a CNOT gate between qubits $c$ and $t$, with the control on qubit $c$. The gates are illustrated in Fig.~\ref{fig:uv}. Observe that $U_{j}$ commutes with $V_{k}$ if $j<k$.
	\begin{figure}
			\centering
		\includegraphics[scale=0.75]{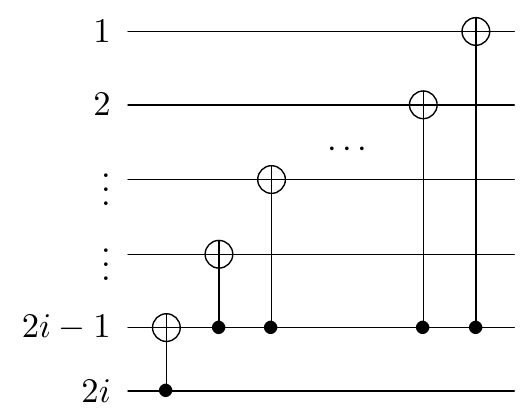}
		~ \includegraphics[scale=0.75]{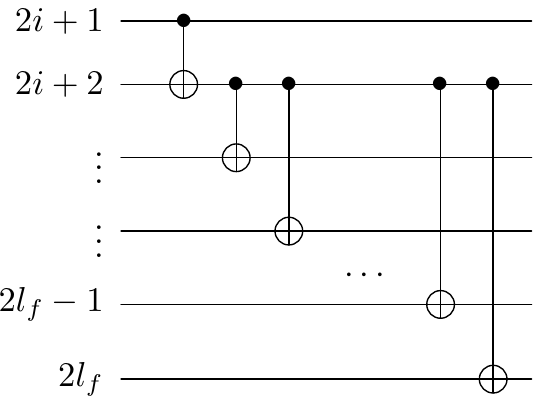}
		\caption{$U_{2i}$ and $V_{2i+1}$}\label{fig:uv}
	\end{figure}
	\begin{lemma}[Properties of $U_{2i}$] 
		\label{lm:u2i}
		The following are some of the properties of $U_{2i}$.
		\begin{subequations}
			\begin{eqnarray}
				U_{2i} Z_{j} U_{2i}^{\dagger} &=& Z_{j} Z_{2i-1} ; j < 2i-1 \label{slm:u2i_1}\\
				U_{2i} Z_{j} U_{2i}^{\dagger} &=& Z_{2i-1} Z_{2i} ; j = 2i-1\label{slm:u2i_2} \\
				U_{2i} Z_{j} U_{2i}^{\dagger} &=& Z_{j} ; j \geq 2i \label{slm:u2i_3}
			\end{eqnarray}
		\end{subequations}
	\end{lemma}
	\begin{proof}
		First we shall prove Eq.~\eqref{slm:u2i_3}. % ~and~\eqref{slm:u2i_5}. 
		Observe that $U_{2i}$ does not act on qubits $j$ if $j > 2i$, therefore $U_{2i}$ commutes with $Z_j$, i.e. $U_{2i}Z_j U_{2i}^\dagger = Z_j U_{2i}U_{2i}^\dagger= Z_j$. If $j=2i$, then 
		\begin{eqnarray*}
			U_{2i}Z_{2i} U_{2i}^\dagger &\overset{(a)}{=} &\prod_{k=1}^{2i-2}\text{CX}^{2i-1}_{k}\text{CX}^{2i}_{2i-1}Z_{2i} \text{CX}^{2i}_{2i-1} \prod_{k=1}^{2i-2}\text{CX}^{2i-1}_{k}\\
			&\overset{(b)}{=}& \prod_{k=1}^{2i-2}\text{CX}^{2i-1}_{k} Z_{2i} \prod_{k=1}^{2i-2}\text{CX}^{2i-1}_{k} \overset{(c)}{=} Z_{2i}.
		\end{eqnarray*}
		where $(a)$ follows from the Eq.~\eqref{eq:u2i}; $(b)$ from the fact that $Z_{2i}$ commutes with the CX gate when it acts on the control qubit and $(c)$ because $Z_{2i}$ commutes with $\text{CX}^{2i-1}_{k}$ for $1\leq k\leq 2i-2$. 
		
		When $j=2i-1$, 
		\begin{eqnarray*}
			U_{2i}Z_{2i-1} U_{2i}^\dagger &= &\prod_{k=1}^{2i-2}\text{CX}^{2i-1}_{k}\text{CX}^{2i}_{2i-1}Z_{2i-1} \text{CX}^{2i}_{2i-1}  \prod_{k=1}^{2i-2}\text{CX}^{2i-1}_{k}\\
			&=& \prod_{k=1}^{2i-2}\text{CX}^{2i-1}_{k} Z_{2i-1} Z_{2i} \prod_{k=1}^{2i-2}\text{CX}^{2i-1}_{k} = Z_{2i-1}Z_{2i}.
		\end{eqnarray*}
		
		Similarly, when $j< 2i-1$, we have the following relations which can be easily verified.
		\begin{eqnarray*}
			U_{2i}Z_{2i} U_{2i}^\dagger &\overset{(a)}{=}& %&\prod_{k=1}^{2i-2}\text{CX}^{2i-1}_{k}\text{CX}^{2i}_{2i-1}Z_{2i} \text{CX}^{2i}_{2i-1} \prod_{k=1}^{2i-2}\text{CX}^{2i-1}_{k}\\
			\prod_{k=1}^{2i-2}\text{CX}^{2i-1}_{k} Z_{j} \prod_{k=1}^{2i-2}\text{CX}^{2i-1}_{k} \\
			&\overset{(c)}{=}& \text{CX}^{2i-1}_{j} Z_{j}  \text{CX}^{2i-1}_{j} = Z_{j} Z_{2i-1}.
		\end{eqnarray*}
	\end{proof}
	\begin{lemma}[Properties of $V_{2i+1}$] 
		\label{lm:v2i}
		The following are some of the properties of $V_{2i+1}$,
		
		\begin{subequations}
			\begin{eqnarray}
				V_{2i+1} Z_{j} V_{2i+1}^{\dagger} &=& Z_{j} ; j \leq 2i+1 \label{slm:v2i_1}\\
				V_{2i+1} Z_{j} V_{2i+1}^{\dagger} &=& Z_{2i+1} Z_{2i+2}; j = 2i+2 \label{slm:v2i_2}\\
				V_{2i+1} Z_{j} V_{2i+1}^{\dagger} &=& Z_{j} Z_{2i+2} ; j > 2i+2\label{slm:v2i_3} %\\
				\end{eqnarray}
		\end{subequations}
		
	\end{lemma}
	\begin{proof}
		The proof is similar to that of Lemma~\ref{lm:u2i} so we omit the details for brevity.
	\end{proof}
	
	With this preparation we are ready to prove the 
	following theorem which gives  the local circuit to transform a color code into two copies of a surface code.

	\begin{figure}
			\centering
		\includegraphics[scale=0.8]{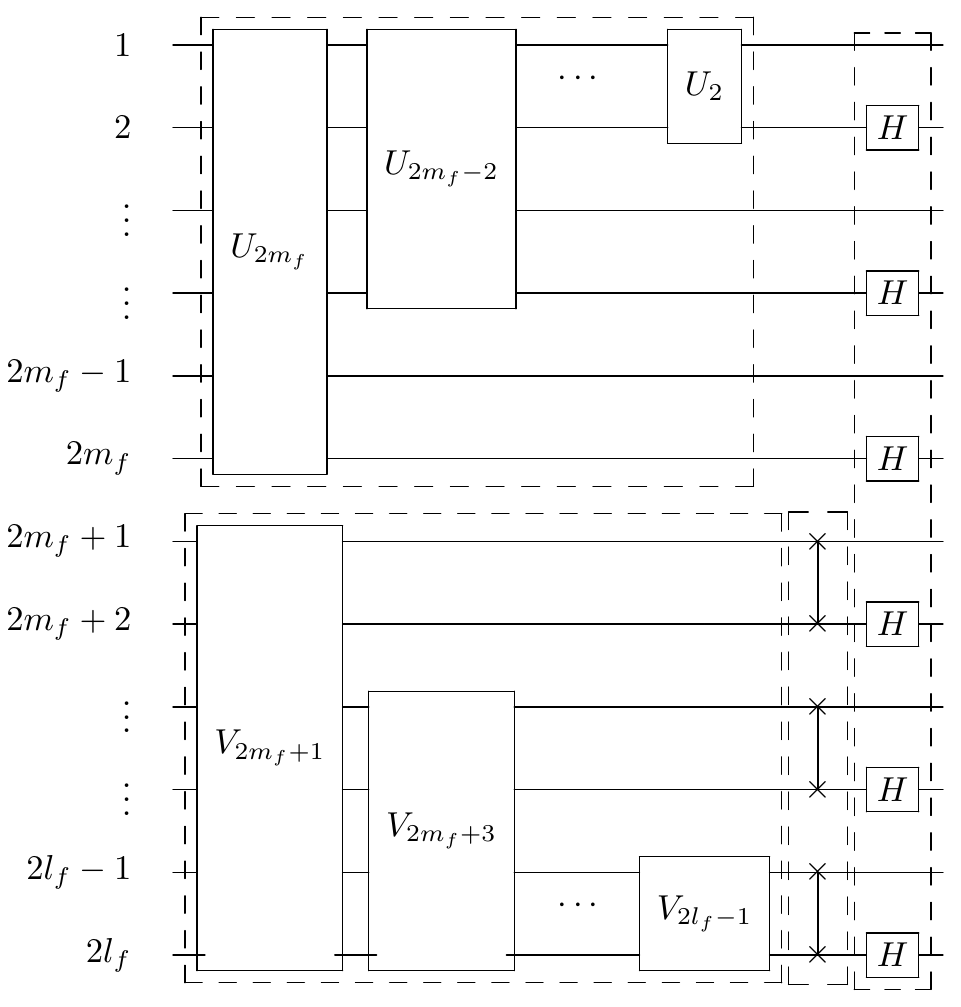}
		\caption{Local (Clifford) circuit to convert a color code into copies of two surface codes. Illustrated for a single 
			$c''$- face. Similar circuits are required for all $c''$-faces.}
		\label{fig:circuit_full}
	\end{figure}

	\begin{theorem}\label{th:ckt}
		The following local Clifford circuit when implemented on each $c''$-face of color code in transforms it into two copies of surface code: % as per Eq.~\eqref{eq:zmap_rec_1}--\eqref{eq:zmap_rec_2}.
		\begin{eqnarray*}
			\prod_{j=1}^{\ell_f} H_{2j} \prod_{j=m_f+1}^{\ell_f} \text{SWAP}_{2j-1}^{2j} V_{2\ell_f-1} \cdots V_{2m_f+1} U_{2} \cdots U_{2m_f}.
		\end{eqnarray*}
		where $H_{i}$ is the single qubit Hadamard gate acting on qubit $i$ and $SWAP_a^b$ is a gate that swaps qubits $a$ and $b$, where the notation is as in Remark~\ref{rm:notation}.
	\end{theorem}
	\begin{proof}
		We need to show that the proposed (local Clifford) circuit transforms $X_i$ and $Z_j$ for $1\leq i,j\leq 2\ell_f$ 
		as per Eqs.~\eqref{eq:zmap_rec_1}--\eqref{eq:xmap-rec06}.
		Recall that if $\{g_1, \ldots g_m, h_1,\ldots, h_m \} $ is a basis for the Pauli group on $m$ qubits such that  $g_ig_j= g_jg_i$, $h_ih_j=h_jh_i$, and $g_ih_j=(-1)^{\delta_{ij}} h_jg_i$, then specifying $g_i$ completely specifies $h_j$ up to phase. Therefore, it suffices to prove that the proposed  circuit realizes the transformation of $Z_i$ as per Eq.~\eqref{eq:zmap_rec_1}--\eqref{eq:zmap_rec_2}. 
		Then, $X_i$ are also transformed as required due to the transformation being Clifford and preserving the commutation relations.
		
		First, observe that $\overline{U} = \prod_{j=1}^{m_f} U_{2j}$ affects $Z_i$ only for $1\leq i \leq 2m_f $ and not for $i>2m_f$ since
		the latter have no support in the first $2m_f$ qubits where $\overline{U}$ acts. If $k=2m_f$, then $\overline{U} Z_{2m_f} \overline{U}^\dagger $ can be written as 
		\begin{eqnarray}
			\overline{U} Z_{2m_f} \overline{U}^\dagger &= &U_2 \cdots U_{2m_f-2} (U_{2m_f} Z_{2m_f} U_{2m_f}^\dagger) U_{2m_f-2}^\dagger \cdots U_2^\dagger \nonumber\\
			&\overset{(a)}{=}& U_2 \cdots U_{2m_f-2} Z_{2m_f} U_{2m_f-2}^\dagger \cdots U_2^\dagger \\
			&\overset{(b)}{=}& Z_{2m_f},
		\end{eqnarray}
		where $(a)$ and $(b)$ follow from repeated application of Eq.~\eqref{slm:u2i_3}. Similarly, for $k=2m_f-1$, we obtain 
		$\overline{U} Z_{2m_f} \overline{U}^\dagger =Z_{2m_f-1}Z_{2m_f}$. 
		In general, if $1\leq k\leq 2m_f $, then by Lemma~\ref{lm:u2i} we obtain 
		\begin{align}
			\overline{U} Z_{2i-1} \overline{U}^\dagger&= Z_{2i}\prod_{j=i}^{m_f} Z_{2j-1}\\
			\overline{U} Z_{2i} \overline{U}^\dagger&=Z_{2i} \prod_{j=i+1}^{m_f} Z_{2j-1} %\\
		\end{align}
		Similarly, $\overline{V} = \prod_{\ell_f}^{i=m_f+1}V_{2i-1}$ affects only $Z_i$ for $2m_f+1\leq i\leq 2\ell_f$
		for $m_f <i\leq \ell_f$. By Lemma~\ref{lm:v2i} we have 
		\begin{eqnarray}
			\overline{V} Z_{2i-1} \overline{V}^\dagger&= Z_{2i} \prod_{j=m_f+1}^{i-1} Z_{2j-1} \\
			\overline{V} Z_{2i} \overline{V}^\dagger&= Z_{2i} \prod_{j=m_f+1}^{i} Z_{2j-1}
		\end{eqnarray}
		
		Next, we apply the SWAP $ \bar{S} = \prod_{j=m_f+1}^{\ell_f} \text{SWAP}_{2j-1}^{2j}$. The \text{SWAP} gates only affect the qubits beyond $2m_f$. 
		\begin{eqnarray}
			Z_{2i}\prod_{j=m_f+1}^{i-1} Z_{2j-1}& \overset{\bar{S}}{\longrightarrow}& Z_{2i-1}\prod_{j=m_f+1}^{i-1} Z_{2j}\\
			Z_{2i}\prod_{j=m_f+1}^{i} Z_{2j-1}  &\overset{\bar{S}}{\longrightarrow}& Z_{2i-1}\prod_{j=m_f+1}^{i} Z_{2j} 
		\end{eqnarray}

		Finally, on applying $\overline{H}=\prod_{i=1}^{\ell_f}H_{2i}$ we obtain \begin{eqnarray}
			Z_{2i}\prod_{j=i}^{m_f} Z_{2j-1} &\overset{\overline{H}}{\longrightarrow}& X_{2i}\prod_{j=i}^{m_f} Z_{2j-1} =\pi(Z_{2i-1})\\
			Z_{2i}\prod_{j=i+1}^{m_f} Z_{2j-1}&\overset{\overline{H}}{\longrightarrow}& X_{2i} \prod_{j=i+1}^{m_f}
			Z_{2j-1} =\pi(Z_{2i})\\
			Z_{2i}\prod_{j=m_f+1}^{i-1} Z_{2j-1}& \overset{\overline{H}}{\longrightarrow}&  \begin{split} X_{2i} \prod_{j=m_f+1}^{i-1} Z_{2j-1}\\=\pi(Z_{2i-1})
			\end{split}\\
			Z_{2i}\prod_{j=m_f+1}^{i} Z_{2j-1} &\overset{\overline{H}}{\longrightarrow} & \begin{split} X_{2i} \prod_{j=m_f+1}^{i} Z_{2j-1} \\ =\pi(Z_{2i}) \end{split}
		\end{eqnarray}
	\end{proof}
	
	The complexity of $U_{2i}$ is $2i-1$ CNOT gates while that of $V_{2i+1}$ is $2\ell_f-2i-1$ CNOT gates. 
	Therefore, the complexity of the circuit in Theorem~\ref{th:ckt} is $m_f^2+(\ell_f-m_f)^2$ CNOT gates, $(\ell_f-m_f)$ 
	SWAP gates and $\ell_f$ H gates for a face with $2\ell_f$ qubits. If we assume that $\ell_f $ is $O(1)$ (i.e. it is a constant) and the number of $c''$-faces is $O(n)$, then the overall complexity of the map is $O(n)$ gates. In particular, for the color code on the square octagonal lattice, $\ell_f=4$ and if we choose $m_f=2$, 
	then the circuit to transform uses $n$ CNOT gates, $n/4$ SWAP gates and $n/2$ H gates. 
	
	We conclude with an example and provide the circuit for a color code on the square octagonal lattice Fig.~\ref{fig:crkt}. The associated $c''$-face of the color code is numbered as in Fig.~\ref{fig:num-qubit} and we assume $2m_f=4$.

	\begin{figure} 
		\centering
		\includegraphics[scale=1]{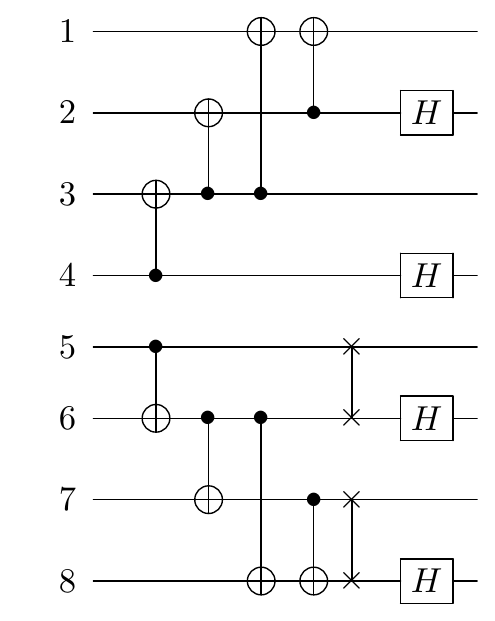}
		\caption{Circuit for single $c''$ face in square octogonal lattice (Fig.~\ref{fig:num-qubit}) with $2m_f=4$ }
		\label{fig:crkt}
	\end{figure}
	
	\section{Conclusion and Outlook }
	In this paper we have shown the local equivalence of the color code to copies of surface codes and applied this equivalence to the decoding of color codes over the bit flip channel and the quantum erasure channel.  Simulation results suggest that to fully exploit the benefits of this map for decoding over the bit flip and quantum erasure channel, we must consider soft decision decoders and/or correlations between the errors. 
Developing such decoders would be a fruitful  direction of further research.
	 Finally, we gave circuits with linear complexity for transforming between color codes and surface codes. 
	 Further optimizations of these circuits would help in reducing the overheads when switching between these codes. 
	 Another possible direction for further research would be to generalize these results to larger alphabet. 
	
    \medskip
    
    \noindent
{\em Acknowledgment.}
AB would like to thank Beni Yoshida and  H\'{e}ctor Bomb\'{i}n  for providing clarifications about their work. 

\appendices
\section{Mapping single qubit errors}
The following figures show the mapping of single qubit errors, see Fig.~\ref{fig:x-error-map},  and erasures, see Fig.~\ref{fig:erasure-map}, for the color code on the square octagon lattice to the surface codes on the square lattice. 

	\begin{figure} 
		\centering
\includegraphics[scale=0.43]{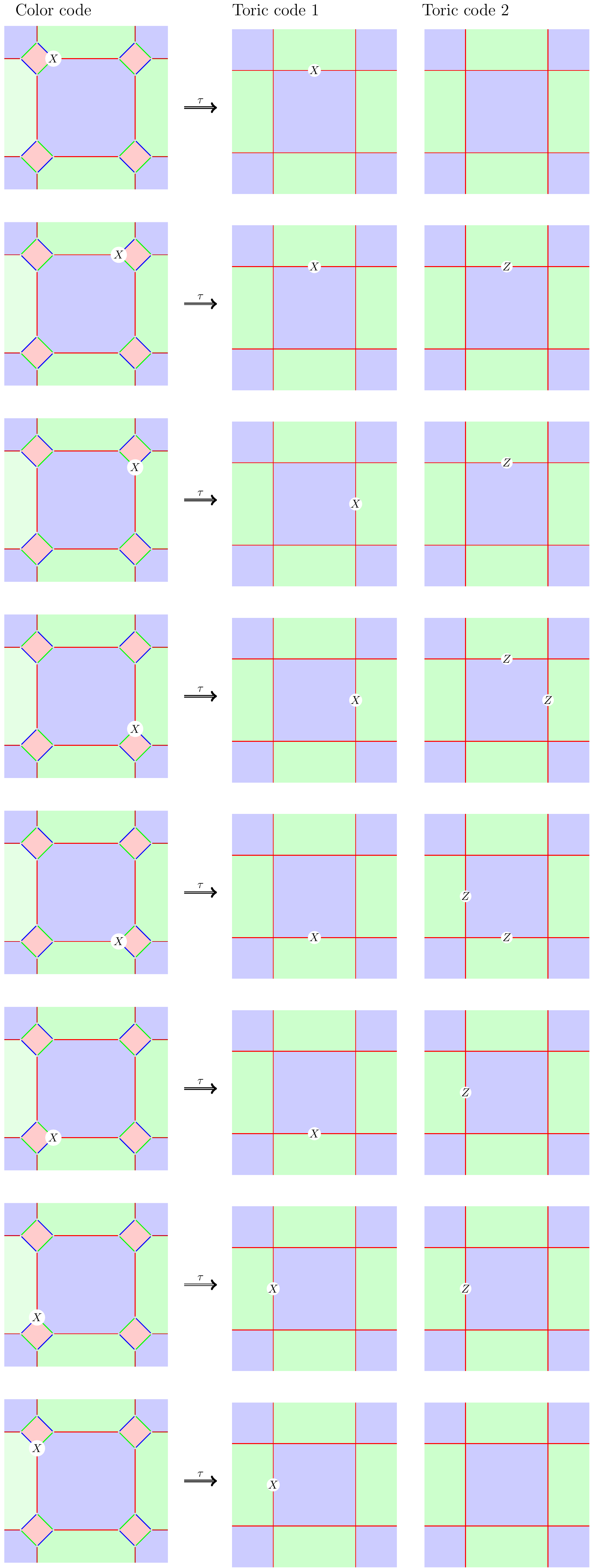}~\hspace{0.5in}~\includegraphics[scale=0.43]{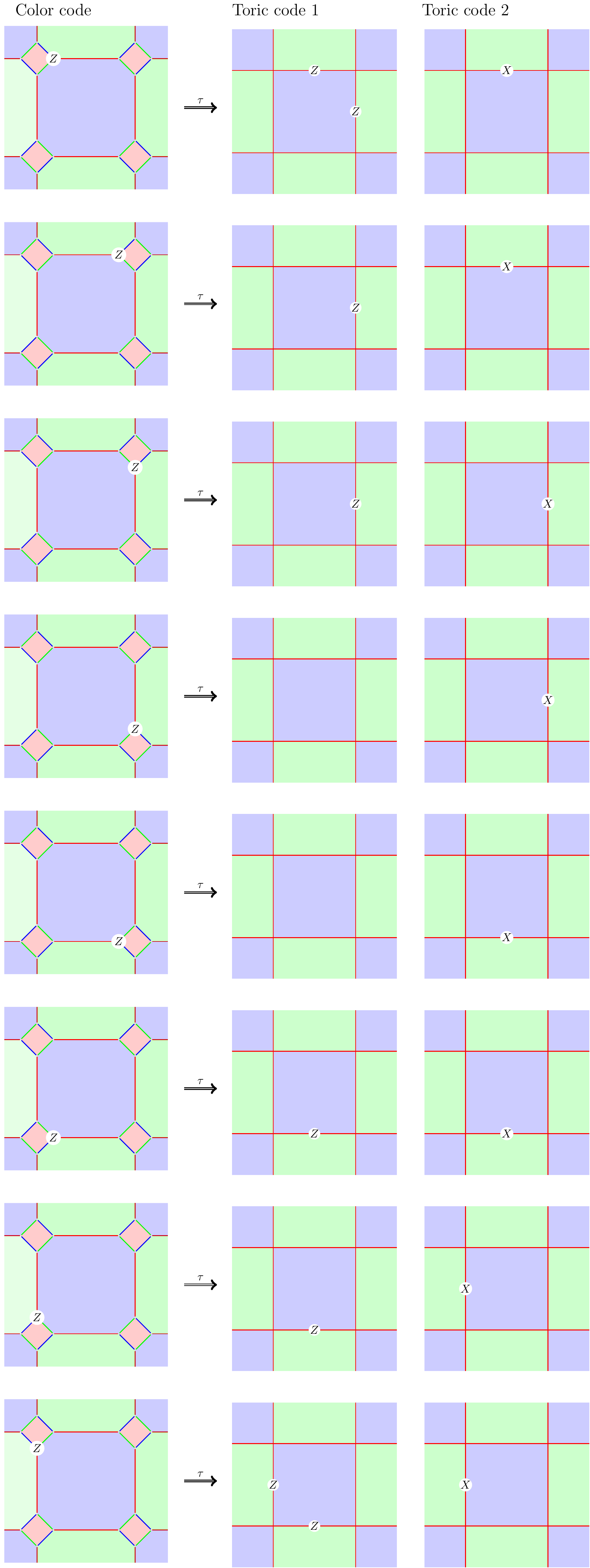}
		\caption{Mapping single qubit $X$ and $Z$ errors for the color code on the square octagon lattice.}\label{fig:x-error-map}
	\end{figure}

%\section{Mapping erasures}
	\begin{figure} 
		\centering
\includegraphics[scale=0.45]{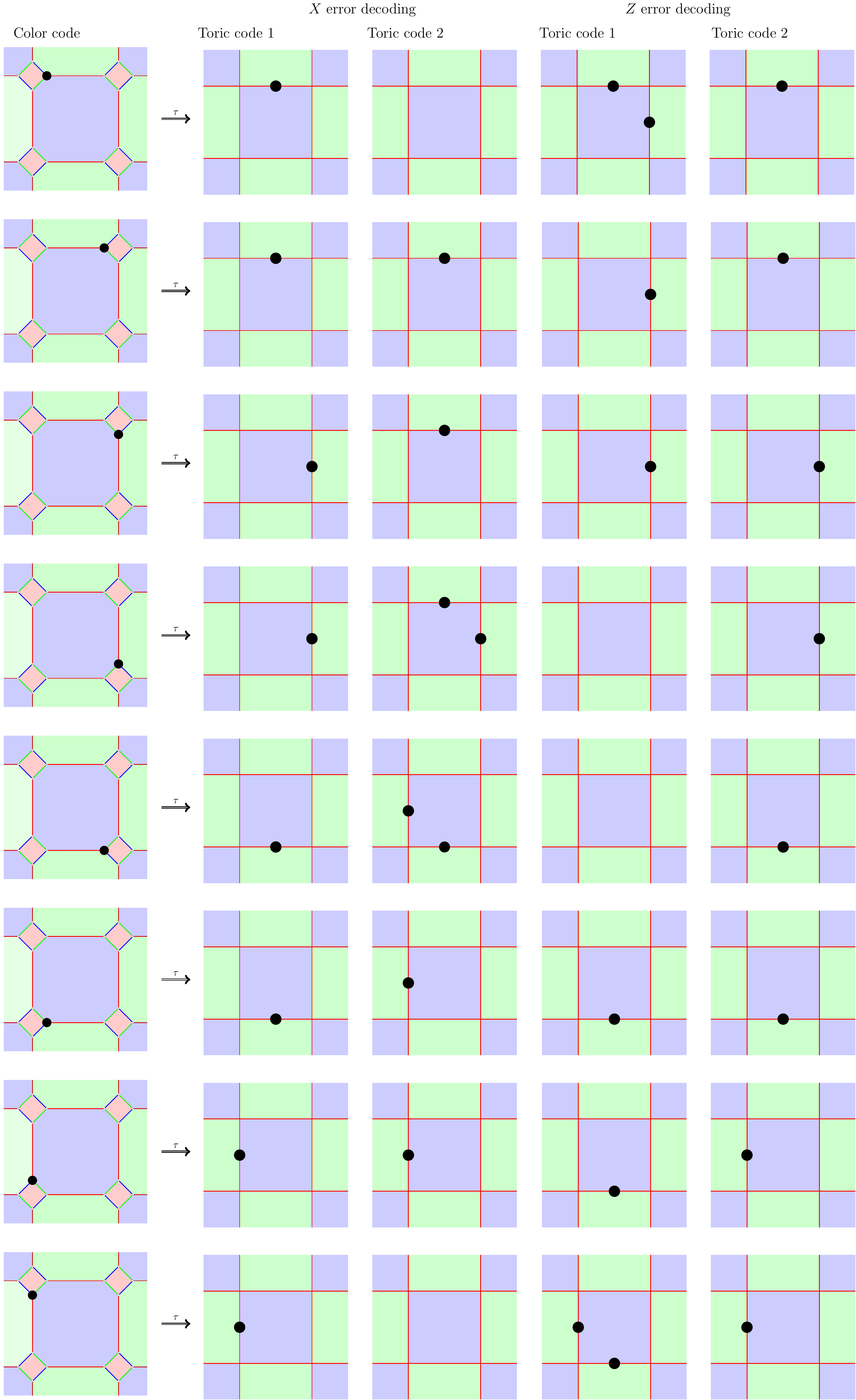}
		\caption{Mapping single qubit erasures for the color code on the square octagon lattice.}\label{fig:erasure-map}
	\end{figure}

%\section*{Acknowledgment}
%This research is supported by a grant from Center for Industrial Consultancy \& Sponsored Research, IIT Madras.
\ifCLASSOPTIONcaptionsoff
  \newpage
\fi



%\bibliographystyle{IEEEtran}

\end{document}